\documentclass[11pt,authoryear]{article}
\pdfoutput = 1
\usepackage{fullpage,amsthm,amsmath,natbib,algorithm,algorithmic,enumitem,afterpage,amssymb,setspace,graphicx,amsfonts,array,verbatim}
\usepackage[hang,flushmargin]{footmisc}
\usepackage[usenames,dvipsnames]{xcolor}
\usepackage[colorlinks=TRUE,citecolor=Blue,linkcolor=BrickRed,urlcolor=PineGreen]{hyperref}

\DeclareMathOperator*{\veco}{vec}

\DeclareMathOperator*{\diag}{diag}
\DeclareMathOperator*{\diverge}{div}
\DeclareMathOperator*{\sure}{SURE}
\DeclareMathOperator*{\gsure}{GSURE}
\DeclareMathOperator*{\mse}{MSE}
\DeclareMathOperator*{\sumo}{Sum}
\DeclareMathOperator*{\Bin}{Bin}
\DeclareMathOperator*{\sign}{sign}

\newtheorem{theorem}{Theorem}

\newtheorem{definition}{Definition}

\allowdisplaybreaks

\newcommand*{\doi}[1]{\href{http://dx.doi.org/#1}{doi: #1}}

\begin{document}
\singlespacing
\title{Adaptive Higher-order Spectral Estimators}
\author{David Gerard$^1$ and Peter Hoff$^{2}$ \\
\\
$^1$Department of Human Genetics,
University of Chicago, Chicago, IL, 60637, USA\\
$^2$Department of Statistical Science,
Duke University, Durham, NC, 27708, USA}
\date{February 22, 2017}
\maketitle
\let\thefootnote\relax\footnotetext{Email: dcgerard@uchicago.edu,  peter.hoff@duke.edu. This research was partially supported by NI-CHD grant R01HD067509.}

\begin{abstract}
  Many applications involve estimation of a signal matrix from a noisy data matrix. In such cases, it has been observed that estimators that shrink or truncate the singular values of the data matrix perform well when the signal matrix has approximately low rank. In this article, we generalize this approach to the estimation of a tensor of parameters from noisy tensor data. We develop new classes of estimators that shrink or threshold the mode-specific singular values from the higher-order singular value decomposition. These classes of estimators are indexed by tuning parameters, which we adaptively choose from the data by minimizing Stein's unbiased risk estimate. In particular, this procedure provides a way to estimate the multilinear rank of the underlying signal tensor. Using simulation studies under a variety of conditions, we show that our estimators perform well when the mean tensor has approximately low multilinear rank, and perform competitively when the signal tensor does not have approximately low multilinear rank. We illustrate the use of these methods in an application to multivariate relational data.\\

\end{abstract}

\section{Introduction}
\label{sec:intro}

Tensor data arise in fields as diverse as relational data
\citep{hoff2015multilinear}, neuroimaging
\citep{zhang2014tensor,li2016parsimonious}, psychometrics
\citep{kiers2001three}, chemometrics
\citep{smilde2005multi,bro2006review}, signal processing
\citep{cichocki2015tensor}, and machine learning
\citep{tao2005supervised}, among others
\citep{kroonenberg2008applied}.
% In particular, tensor decompositions \citep{kolda2009tensor}, as a
% way of identifying tensor structure, are being applied to each of
% these fields.
A tensor $\mathcal{X} \in \mathbb{R}^{p_1\times\cdots\times p_K}$ with
$p_k \in \{1,2,\ldots\}$ of order $K$ is a $K$-way array where the
elements $\mathcal{X}_{[i_1,\ldots,i_K]}$ are indexed by $i_k \in
\{1,2,\ldots,p_k\}$ for $k = 1,\ldots,K$. For example, a multivariate
relational dataset can be expressed as a tensor, where element
$\mathcal{X}_{[i,j,t]}$ of the tensor is the $t$th relation between
actors $i$ and $j$.

Often, a tensor is corrupted by noise. The model we consider for this
is:
\begin{align}
  \label{equation:normal.model}
  \mathcal{X} = \Theta + \mathcal{E},\text{  } \mathcal{E}_{[i_1,\ldots,i_K]} \sim N(0,\tau^2) \text{ independent for } i_k = 1,\ldots,p_k, \text{ and } k = 1,\ldots,K,
\end{align}
where $\Theta \in \mathbb{R}^{p_1\times\cdots\times p_K}$ is the
signal and $\mathcal{E} \in \mathbb{R}^{p_1\times\cdots\times p_K}$ is
the additive Gaussian measurement error or noise with mean 0 and
various $\tau^2$. The performance of an estimator
$t(\mathcal{X})\in \mathbb{R}^{p_1\times\cdots\times p_K}$ can be
evaluated by statistical risk under quadratic loss, i.e.\ mean squared
error (MSE):
\begin{align}
  \label{equation:mse}
  \mse(t(\mathcal{X})) = E_{\Theta}[||\Theta - t(\mathcal{X})||^2] = \sum_{\mathbf{i}} E_{\Theta}[(\Theta_{[\mathbf{i}]} - t(\mathcal{X})_{[\mathbf{i}]})^2],
\end{align}
where $\mathbf{i} = (i_1,\ldots,i_K)$ is a $K$-tuple of tensor
indices.

In the matrix variate case, $X \in \mathbb{R}^{p \times n}$, an
investigator often believes that the mean is well approximated by a
low rank matrix. There has been much work on ``denoising'' (or mean
estimation) in matrix variate data by using this knowledge. A typical
estimation scheme begins by computing the singular value decomposition
(SVD) of $X$:
\begin{align}
  \label{equation:svd}
  X = UDV^T,
\end{align}
where, in the case $n \geq p$, $U \in \mathbb{R}^{p \times p}$ is
orthogonal, $D = \diag(\sigma_1,\ldots,\sigma_p)$ with $\sigma_1 \geq
\ldots\geq \sigma_p \geq 0$, and $V \in \mathbb{R}^{n \times p}$
contains orthonormal columns. The columns of $U$ and $V$ are,
respectively, the left and right singular vectors of $X$ and the
diagonal elements of $D$ are the singular values. A key property of
the SVD is that the number of non-zero singular values of $X$ is
precisely the rank of $X$. One widely studied approach to estimating
$\Theta$ when it is assumed that $\Theta$ has nearly low rank is to
shrink the singular values of $X$ towards $0$ while keeping the
singular vectors unchanged, thereby inducing an (approximately) low
rank estimate. The resulting ``spectral'' estimator $t(\mathcal{X})$
of $\Theta$ then takes the form $t(\mathcal{X}) = Uf(D)V^T$ where
$f(D) = \diag(f_1(\sigma_1),\ldots,f_K(\sigma_K))$ and each
$f_i(\cdot)$ shrinks the singular values towards $0$. These estimators
are orthogonally equivariant, meaning that $t(WXZ^T) = Wt(X)Z^T$ for
orthogonal matrices $W,Z$ \citep{shabalin2013reconstruction}.

Early work on singular value shrinkage estimation from a
non-statistical perspective began with \cite{eckart1936approximation},
where they proved that the best rank $r$ approximation to the data
matrix $X \in \mathbb{R}^{p \times n}$ (in terms of sum of squared
differences from $X$) is found with the shrinkage function:
\begin{align}
  \label{equation:eckart.young}
  f_i(\sigma_i) = \sigma_i 1(i \leq r),
\end{align}
where $1(\cdot)$ is the indicator function. We call
(\ref{equation:eckart.young}) the truncation estimator. However,
approximating the data $X$ well is not the same as estimating the
underlying signal $\Theta$ well. In terms of estimating $\Theta$, the
matrix $X$ is unbiased, minimax, and the maximum likelihood estimator
under normally distributed errors. However, it is well known that
shrinkage estimators, such at that of \cite{stein1981estimation} can
uniformly dominate $X$ in terms of risk. This seminal shrinkage
estimator, in the context of matrix estimation, is given by
\begin{align}
  \label{equation:stein}
  f_i(\sigma_i) = \left(1 - \frac{\lambda}{\sum_{i = 1}^p\sigma_i^2}\right)\sigma_i,
\end{align}
where $\lambda > 0$ is some tuning parameter. For data that exhibit
associations between the rows and/or columns of the mean matrix, the
estimator of \cite{efron1972empirical}, given by
\begin{align}
  \label{equation:efron.morris}
  f_i(\sigma_i) = \sigma_i - \frac{\lambda}{\sigma_i},
\end{align}
was introduced and results in different amounts of shrinkage for each
singular value.  \cite{efron1976multivariate} improved upon this
estimator with a generalization of both (\ref{equation:stein}) and
(\ref{equation:efron.morris}), given by
\begin{align}
  \label{equation:improved.em}
  f_i(\sigma_i) = \left(1 - \frac{\gamma}{\sum_{i = 1}^p\sigma_i^2}\right)\sigma_i - \frac{\lambda}{\sigma_i},
\end{align}
where $\lambda > 0$ and $\gamma > 0$ are tuning parameters.

More recent work has focused on estimators whose functions
$f_i(\cdot)$ induce sparsity in the singular values, which may be more
appropriate than (\ref{equation:stein}),
(\ref{equation:efron.morris}), and (\ref{equation:improved.em}) in
cases where the true signal itself has (approximately) low
rank. Motivated by penalized maximum likelihood estimation, the
hard-thresholding estimator
\begin{align}
  \label{equation:hard.thresholding}
  f_i(\sigma_i) = \sigma_i 1(\sigma_i \geq \lambda)
\end{align}
and the soft-thresholding estimator
\begin{align}
  \label{equation:soft.thresholding}
  f_i(\sigma_i) = (\sigma_i - \lambda)_+
\end{align}
were introduced \citep[for example]{candes2013unbiased}. Here, $(y)_+
= \max(y,0)$ is the ``positive part'' function. A clever shrinkage
function that includes (\ref{equation:hard.thresholding}),
(\ref{equation:soft.thresholding}), and a truncated version of
(\ref{equation:efron.morris}) \citep{verbanck2015regularised} as
special cases is that of \cite{josse2015adaptive}:
\begin{align}
  \label{equation:josse.sardy.est}
  f_i(\sigma_i) = \sigma_i\left(1 - \frac{\lambda^{\gamma}}{\sigma_i^{\gamma}}\right)_+.
\end{align}
This estimator was inspired by the adaptive LASSO
\citep{zou2006adaptive}. A variety of other shrinkage estimators have
also been developed
\citep{nadakuditi2014optshrink,shabalin2013reconstruction}.

All of these estimators are specific to matrix-variate data. If one
were to apply these matrix methods to a tensor, one would first
convert the tensor into a matrix. For a $K$-dimensional tensor, such
``matricization'' destroys the indexing structure along all but one of
the dimensions. This may be detrimental to estimation if, in addition
to a data set having approximately low rank, it also has approximately
low \emph{multilinear} rank (see Section \ref{sec:tensor}), that is,
``matricizing'' along each index set, or ``mode'', results in a low
rank matrix.

An extreme simulated example that exhibits this phenomenon is
presented in Figure \ref{fig:sim.extreme}. There, we plotted the
mode-specific singular values of a tensor that we generated to have
full rank along one mode and low ranks along two modes. That is, we
plotted the singular values of each matricization of the tensor. If an
analyst were presented with a noisy version of this tensor and only
matricizing along the first mode, then they would only observe a noisy
realization of the solid lines, which would suggest the data are full
rank. However, the second and third modes have low-rank structure and
shrinking the singular values along these additional modes may improve
estimation.

\begin{figure}
\begin{center}
\includegraphics{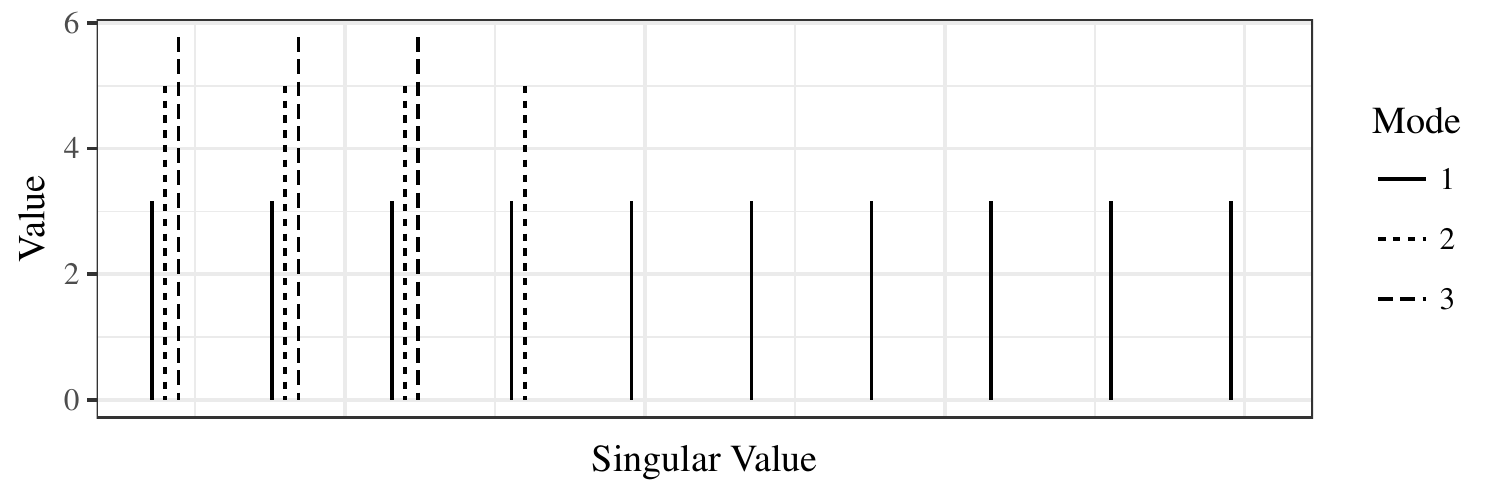}
\caption{Mode-specific singular values of simulated tensor with full
  rank along first mode and low-ranks along second and third modes.}
\label{fig:sim.extreme}
\end{center}
\end{figure}

% \textbf{New data showing low multilinear rank. Show scree plot for
%   each mode.} If we were only matricizing along the time dimension we
% would only observe the dotted lines, which seem to suggest the data
% have approximately rank 1. However, the vertical and horizontal
% dimensions also seem to have approximately low rank, as evidenced by
% the solid and dashed lines. Shrinking the singular values along these
% additional modes may also improve estimation.

%% In general, having low rank along one mode does not imply having
%% low rank along all modes. It is possible to have different
%% mode-specific ranks, and possible to shrink toward low multilinear
%% rank. For example, in the left plot of Figure \ref{fig:mario.movie}
%% we simulated a $10 \times 10 \times 10$ dimensional tensor data set
%% that contains full rank along the first mode and low rank along
%% modes 2 and 3. If we were only to observe the first mode matrix
%% unfolding, then we would not be aware of the low rank structure,
%% when in fact there is extremely low multilinear rank structure.

In this article, we introduce a family of estimators that shrink
tensor-valued data towards having (approximately) low multilinear
rank. We perform this shrinkage on a reparameterization of the
higher-order singular value decomposition (HOSVD) of
\cite{de2000multilinear}, where we shrink the mode-specific singular
values of the data tensor towards zero. We consider classes of such
``higher-order spectral estimators'', where a class is defined by a
mode-specific shrinkage function indexed by a tuning parameter. We
propose to adaptively select the tuning parameters by minimization of
an unbiased estimate of the risk.

%% Shrinkage estimators often contain tuning parameters whose
%% selection is difficult process. In the case of matrix spectral
%% estimators, others have chosen the amount of shrinkage by minimax
%% considerations \citep{efron1972empirical,stein1981estimation},
%% cross-validation
%% \citep{bro2008cross,owen2009bi,josse2012selecting}, asymptotic
%% considerations \citep{gavish2014optimalhard,gavish2014optimal}, and
%% recently, by minimizing an unbiased estimate of the MSE
%% \citep{candes2013unbiased,josse2015adaptive}. For our second
%% contribution, we extend this last method to the tensor case by
%% deriving Stein's unbiased risk estimate (SURE) for higher-order
%% spectral estimators.

Our paper is organized as follows. In Section \ref{sec:tensor}, we
review tensors and the HOSVD. We then present how one may define
functions that shrink the mode-specific singular values of the
HOSVD. In particular, we present two specific estimators that shrink
the data tensor towards having (approximately) low multilinear rank
and provide some discussion on the intuition behind these
estimators. In Section \ref{sec:sure_form}, we review Stein's unbiased
risk estimates (SURE), then derive the SURE for a broad class of
higher-order spectral estimators. In Section \ref{sec:simulation} we
present simulations demonstrating that (1) tensor specific methods
perform better when the mean tensor has approximately low multilinear
rank; (2) when the mean tensor has low multilinear rank our methods
accurately estimate the multilinear rank; and (3) tensor specific
methods perform competitively when the signal tensor does not have
approximately low multilinear rank. In Section \ref{sec:NBA} we
illustrate the use of these methods in an application to multivariate
relational data. We finish with a discussion in Section
\ref{sec:discussion}.

\section{The higher-order SVD and higher-order spectral estimators}
\label{sec:tensor}

Some tensor data sets have approximately low \emph{multilinear rank},
which we now define. Recall that the rank of a matrix is the dimension
of the vector space spanned by its columns and rows. Define the
$k$-mode vectors of a tensor $\mathcal{X} \in
\mathbb{R}^{p_1\times\cdots\times p_K}$ as the $p_k$-dimensional
vectors formed from $\mathcal{X}$ by varying $i_k$ and keeping the
other indices fixed. The $k$-mode rank $r_k$ is the dimension of the
span of the $k$-mode vectors, and the multilinear rank of the
$K$-order tensor $\mathcal{X}$ is the $K$-tuple,
$(r_1,\ldots,r_K)$. Define the $k$-mode matricization
\citep{kolda2009tensor}, or $k$-mode unfolding, of $\mathcal{X}$ to be
$\mathcal{X}_{(k)}\in \mathbb{R}^{p_k \times p/p_k}$ (with $p =
\prod_{k=1}^Kp_k$) where element $(i_1,\ldots,i_K)$ in $\mathcal{X}$
maps to element $(i_k,j)$ in $\mathcal{X}_{(k)}$ where
\begin{align*}
  j = 1 + \sum_{\substack{n = 1 \\ n\neq k}}^{K}(i_n - 1)J_n \text{ with } J_n = \prod_{\substack{ m = 1 \\ m \neq k}}^{n-1}p_m.
\end{align*}
Then, equivalently, $r_k$ is the rank of $\mathcal{X}_{(k)}$.

The SVD , presented in Section \ref{sec:intro}, has been used to
shrink matrix valued data towards low rank. One generalization of the
SVD to tensors is the HOSVD of \cite{de2000multilinear}, which relates
directly to multilinear rank.

\begin{definition}[HOSVD of \cite{de2000multilinear}]
  Let $\mathcal{X}_{(k)} = U_kD_kV_k^T$ be the SVD of each $k$-mode
  unfolding of $\mathcal{X}$. Let $\mathcal{S} =
  (U_1^T,\ldots,U_K^T)\cdot\mathcal{X}$, then
  \begin{align}
    \label{equation:hosvd}
    \mathcal{X} = (U_1,\ldots,U_K)\cdot \mathcal{S}
  \end{align}
  is the higher-order singular value decomposition (HOSVD).
\end{definition}
The product ``$\cdot$'' in (\ref{equation:hosvd}) between a list of
matrices, $\{U_1,\ldots,U_K\}$ for $U_k \in \mathbb{R}^{p_k \times
  p_k}$, and a tensor, $\mathcal{S} \in
\mathbb{R}^{p_1\times\cdots\times p_K}$ is called the \emph{Tucker
  product}. The Tucker product is defined through the $k$-mode
matricizations of $(U_1,\ldots,U_K)\cdot \mathcal{S}$:
\begin{align*}
  \begin{split}
    &\mathcal{X} = (U_1,\ldots,U_K)\cdot \mathcal{S}\\
    &\Leftrightarrow \mathcal{X}_{(k)} = U_k\mathcal{S}_{(k)}(U_K^T\otimes\cdots\otimes U_{k+1}^T\otimes U_{k-1}^T\otimes\cdots\otimes U_1^T) = U_k\mathcal{S}_{(k)}U_{-k}^T,
  \end{split}
\end{align*}
where ``$\otimes$'' is the Kronecker product.  The ``core array'',
$\mathcal{S}$ has the property of \emph{all-orthogonality} where
\begin{align*}
  \mathcal{S}_{(k)}\mathcal{S}_{(k)}^T = D_k^2 \text{ for all } k = 1,\ldots,K.
\end{align*}

The HOSVD is multilinear rank-revealing in the same way the SVD is
rank-revealing. That is, let $D_k =
(\mathcal{S}_{(k)}\mathcal{S}_{(k)}^T)^{1/2} =
\diag(\sigma_1^k,\ldots,\sigma_{p_k}^k)$ be the mode specific singular
values of $\mathcal{X}$. Then the multilinear rank of $\mathcal{X}$ is
$(r_1,\ldots,r_K)$ if $D_k$ contains $r_k$ non-zero mode-specific
singular values. In the core array, this is equivalent to
$\mathcal{S}$ containing zeros everywhere except in one of the
``corners'': $\mathcal{S}_{[1:r_1,\ldots,1:r_K]}$, where $1:r_k =
1,\ldots,r_k$. It is possible, then, to shrink $\mathcal{S}$ towards
having (approximately) low multilinear rank by shrinking the elements
in $\mathcal{S}$ towards $0$. We propose doing this via a
re-parameterization of $\mathcal{S}$, given as follows:
\begin{align}
  \label{equation:hosvd.rewrite}
  \begin{split}
    \mathcal{X} = (U_1,\ldots,U_K) \cdot (D_1,\ldots,D_K)\cdot (D_1^{-1},\ldots,D_K^{-1})\cdot \mathcal{S}\\
    = (U_1,\ldots,U_K) \cdot (D_1,\ldots,D_K)\cdot \mathcal{V},
  \end{split}
\end{align}
where $\mathcal{S} = (D_1,\ldots,D_K)\cdot \mathcal{V}$. Our
higher-order spectral estimators shrink $\mathcal{S}$ by shrinking
each mode-specific $D_k$. We abuse notation a little by allowing
``$\cdot$'' to also represent a binary operator between two lists of
matrices whose operation is component-wise multiplication. This should
not cause confusion because $(A_1B_1,\ldots,A_KB_K) \cdot \mathcal{C}
= (A_1,\ldots,A_K) \cdot [(B_1,\ldots,B_K) \cdot \mathcal{C}]$.

Using reparameterization (\ref{equation:hosvd.rewrite}), we now define
higher-order spectral estimators of $\Theta$ under the model
(\ref{equation:normal.model}).
\begin{definition}
  Let $\mathcal{X} = (U_1,\ldots,U_K) \cdot (D_1,\ldots,D_K)\cdot \mathcal{V}$ as in (\ref{equation:hosvd.rewrite}) with $D_k = \diag(\sigma_1^k,\ldots,\sigma_{p_k}^k)$. An estimator $t(\mathcal{X})$ of the form
  \begin{align}
    \label{equation:ho.spect.est}
    t(\mathcal{X}) = (U_1,\ldots,U_K) \cdot (f^1(D_1),\ldots,f^K(D_K))\cdot \mathcal{V},
  \end{align}
  where $f^k(D_k) =
  \diag(f_1^k(\sigma_1^k),\ldots,f_{p_k}^k(\sigma_{p_k}^k))$, is
  called a \emph{higher-order spectral estimator}.
\end{definition}

Each of the matrix shrinkage functions listed in Section
\ref{sec:intro}
(\ref{equation:eckart.young})-(\ref{equation:josse.sardy.est}) may, in
principle, be applied to each mode in our higher-order spectral
estimator (\ref{equation:ho.spect.est}). We focus on two examples of
higher-order spectral estimators. One of these is a generalization of
the matrix truncation estimator (\ref{equation:eckart.young}) and the
other is a generalization of the matrix soft-thresholding estimator
(\ref{equation:soft.thresholding}). The former can be used to choose
the multilinear rank of $\Theta$, the latter is for estimation of
$\Theta$ when we suspect that the mean tensor has approximately low
multilinear rank.

\paragraph{Example: Truncated HOSVD to find the multilinear rank.}
The first step in many tensor applications is to choose the
multilinear rank of the underlying signal, a difficult task
\citep{timmerman2000three,kiers2003fast,ceulemans2006selecting}. The
methods in this paper present a way to choose the multilinear
rank. The truncated HOSVD is one popular way to induce low multilinear
rank \citep{de2000multilinear}. Given multilinear rank
$(r_1,\ldots,r_K)$, it is found by taking the HOSVD
(\ref{equation:hosvd}) and setting all elements in $\mathcal{S}$
except the ``corner'' $\mathcal{S}_{[1:r_1,\ldots,1:r_K]}$ to $0$. The
truncated HOSVD may be viewed as a higher-order spectral estimator
(\ref{equation:ho.spect.est}), where
\begin{align}
  \label{equation:trunc.shrink}
  f^k_i(\sigma_i^k) = \sigma_i^k1(i \leq r_k).
\end{align}
This sets to 0 all but $r_k$ of the mode-specific singular values,
resulting in an estimate of $\Theta$ that has multilinear rank
$(r_1,\ldots,r_K)$. The set of all possible multilinear ranks defines
a class of reduced rank estimators of $\Theta$. In this paper, we
suggest adaptively selecting an estimator from this class by
minimizing an unbiased estimate of the risk.

\paragraph{Example: Mode-specific soft-thresholding.}
Shrinking all of the singular values can generally improve estimation
over just truncating the smallest few singular values. A popular form
of shrinkage that accomplishes this, a result of nuclear-norm
regularization, is the soft-thresholding estimator
(\ref{equation:soft.thresholding}). The second estimator we explore is
obtained by applying soft-thresholding to the mode-specific singular
values:
\begin{align}
  \label{equation:mode.specific.soft}
  f^k_i(\sigma_i^k) = (\sigma_i^k - \lambda_k)_+.
\end{align}
As with the previous example, the set of
$(\lambda_1,\ldots,\lambda_K)$ defines a class of estimators. We
propose adaptively selecting a member of this class by minimizing an
unbiased estimate of the risk.

A few words are in order about the mode-specific soft-thresholding
estimator in (\ref{equation:mode.specific.soft}). First, we note that
the resulting core array $(f^1(D_1)D_1^{-1},\ldots,f^K(D_K)D_K^{-1})
\cdot \mathcal{S}$ is not generally all-orthogonal. Hence, the
$f^k(D_k)$ are not actually the new mode-specific singular values of
the estimator $t(\mathcal{X})$. That is, it would be incorrect to
think that subtracting off $\lambda_1$ from the first-mode singular
values means that the new first-mode singular values are
$\sigma_{i_1}^1 - \lambda_1$. We are altering the mode-specific
singular values, but the relationship is complex. Rather, the proper
intuition for shrinkage functions of the form
(\ref{equation:mode.specific.soft}) is that the larger the value of
$\lambda_k$, the more dispersed the resulting mode-specific singular
values tend to be on a normalized scale. Likewise, the more negative
the value of $\lambda_k$ to the singular values the less dispersed
the resulting mode-specific singular values tend to be. To gain
intuition regarding this phenomenon, we provide an extreme case. We
generated a $10 \times 10 \times 10$ tensor where each mode had
approximately the same singular values. The first-mode specific
singular values were $(947, 873, 844, 801, 746, 698, 675, 597, 524,
472)$. We applied the mode specific soft-thresholding function
(\ref{equation:mode.specific.soft}) to each mode with $\lambda_1 =
500$, $\lambda_2 = 0$, $\lambda_3 = -10000$. We then calculated the
mode-specific singular values of the resulting tensor and compared
these to the original mode-specific singular values, scaled to sum to
one. The comparisons can be found in Figure
\ref{fig:function.intuition}. The changed (and normalized) singular
values are more dispersed for the first mode, remain relatively
unchanged for the second, and are less dispersed for the third.

We have found that we can improve performance (with respect to MSE) by
adding an overall scale tuning parameter. That is, we consider a
shrinkage estimator of the form:
\begin{align}
  \label{equation:msst.est}
  t(\mathcal{X}) = c\ (U_1,\ldots,U_K)\cdot(f^1(D_1)D_1^{-1},\ldots,f^K(D_K)D_K^{-1})\cdot\mathcal{S},
\end{align}
where $c > 0$ is the overall scale parameter,
$f^k(D_k) =
\diag(f_1^k(\sigma_1^k),\ldots,f_{p_k}^k(\sigma_{p_k}^k))$,
and $f_i^k(\cdot)$ is from (\ref{equation:mode.specific.soft}).

\begin{figure*}
\begin{center}
\includegraphics{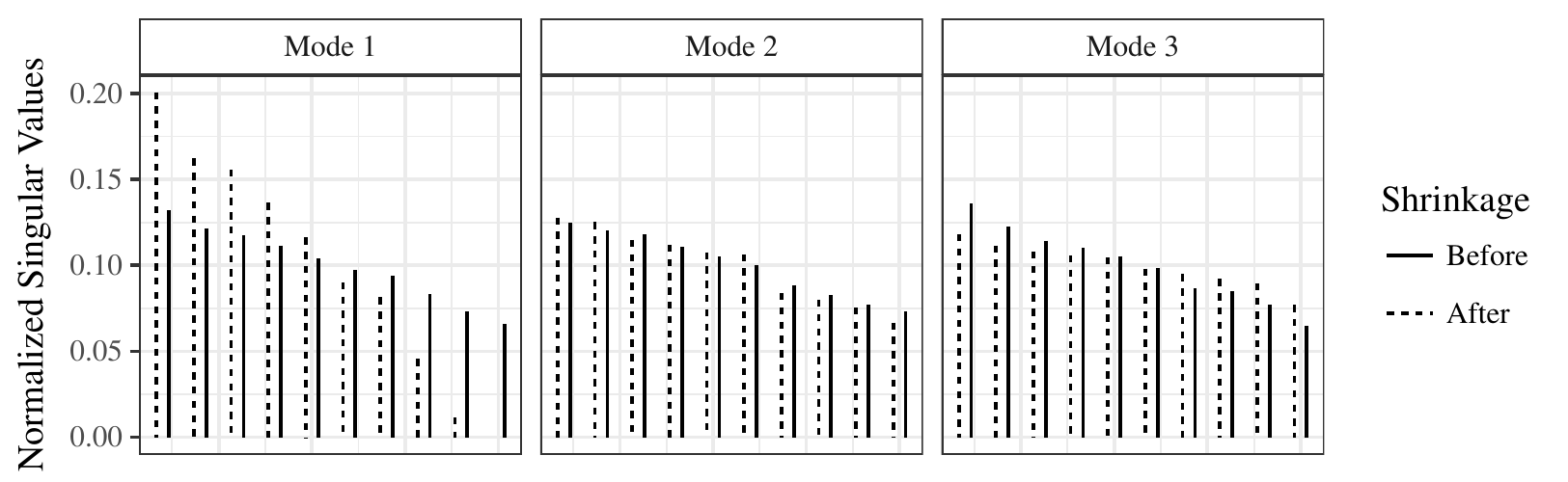}
\caption{Singular values for the three modes, before and after
  shrinkage, normalized to sum to one.}
\label{fig:function.intuition}
\end{center}
\end{figure*}

\section{Stein's unbiased risk estimate}
\label{sec:sure_form}

Both shrinkage function (\ref{equation:trunc.shrink}) and
(\ref{equation:msst.est}) define classes of estimators, indexed by
tuning parameters.  Ideally, we would like to choose these tuning
parameters by minimizing the risk (\ref{equation:mse}). However,
because the mean $\Theta$ is unknown, minimization of
(\ref{equation:mse}) with respect to the tuning parameters is not
possible. One approach for selecting an estimator from one of these
classes is to minimize a risk estimate that does not depend on the
unknown parameter. One such estimate is Stein's unbiased risk
estimate:
\begin{theorem}[\cite{stein1981estimation}]
  \label{theorem:stein.theorem}
  Under the model (\ref{equation:normal.model}), suppose $t: \mathbb{R}^{p_1\times\cdots\times p_K} \rightarrow \mathbb{R}^{p_1\times\cdots\times p_K}$ is an almost differentiable function for which
  \begin{align}
    \label{equation:int.condition}
    E_{\Theta}\left[\sum_{\mathbf{i}} \left|\frac{d}{d\mathcal{X}_{[\mathbf{i}]}}t_{\mathbf{i}}(\mathcal{X}_{[\mathbf{i}]})\right|\right] < \infty.
  \end{align}
  Then
  \begin{align*}
    \mse(t(\mathcal{X})) = E_{\Theta}[||\Theta - t(\mathcal{X})||^2] = E_{\Theta}\left[ ||t(\mathcal{X}) - \mathcal{X}||^2 + 2\tau^2\diverge(t(\mathcal{X})) - p\tau^2\right],
  \end{align*}
  where $\diverge(\cdot)$ is the divergence of $t(\cdot)$. We denote
  Stein's unbiased risk estimate (SURE) as
  \begin{align}
    \label{equation:sure.form}
    \sure(t) = ||t(\mathcal{X}) - \mathcal{X}||^2 + 2\tau^2\diverge(t(\mathcal{X})) - p\tau^2.
  \end{align}
\end{theorem}

``Almost differentiable'' basically means differentiable everywhere
except on a set of Lebesgue measure zero \citep[Definition
1]{stein1981estimation}. Because the SURE (\ref{equation:sure.form})
does not depend on the parameter values $\Theta$, we can minimize the
SURE and use this minimization as a proxy for minimizing the risk. In
many cases, adaptive estimators obtained by minimizing SURE over a
class of estimators yields improved risk performance, as was observed
by \cite{candes2013unbiased} in the matrix case.

The difficult part of (\ref{equation:sure.form}) is calculating the
divergence. We will spend the next two subsections performing this
task. First, we will calculate the differentials for the elements of
the altered HOSVD (\ref{equation:hosvd.rewrite}) in Subsection
\ref{subsection:diff}. Then we will use these differentials to derive
the divergence of estimators of the form (\ref{equation:ho.spect.est})
in Subsection \ref{subsection:div}. This divergence can then be
inserted into (\ref{equation:sure.form}) to obtain the SURE.

\subsection{Differentials of the HOSVD}
\label{subsection:diff}
In this subsection, we calculate the differentials for the elements in
the altered HOSVD (\ref{equation:hosvd.rewrite}). In what follows, we
will assume that $\mathcal{X}$ has full multilinear rank. Given that
$p_k \leq p/p_k$ for all $k = 1,\ldots,K$, where $p =
\prod_{k=1}^Kp_k$, this rank condition is fulfilled almost surely for
data $\mathcal{X}$ that have a p.d.f. that is absolutely continuous
with respect to Lebesgue measure on $\mathbb{R}^{p_1\times\cdots\times
  p_K}$ \cite[Proposition 7.2]{de2008tensor}.

\begin{theorem}
  \label{theorem:diff}
  The differentials of $D_k$, $U_k$, and $\mathcal{V}$ from
  (\ref{equation:hosvd.rewrite}) are given in equations
  (\ref{equation:d.D}), (\ref{equation:d.U}), and
  (\ref{equation:d.V}), respectively.
\end{theorem}

An outline of the derivation is as follows: Because each $U_k$ and
$D_k$ from the HOSVD is from the SVD of $\mathcal{X}_{(k)} =
U_kD_kV_k^T$, the calculation begins by recognizing that the
differentials of the $U_k$'s and the $D_k$'s are the same as in the
matrix case. The differentials can then be re-written as functions of
the terms in the HOSVD. To obtain the differential of $\mathcal{V}$,
we write $\mathcal{X} =
(U_1,\ldots,U_K)\cdot(D_1,\ldots,D_K)\cdot\mathcal{V}$ and apply the
chain rule to each $U_k$, each $D_k$, then to $\mathcal{V}$. We then
solve for the differential of $\mathcal{V}$, which may be written in
terms of the differentials of the $U_k$'s and the $D_k$'s.

\begin{proof}[Proof of Theorem \ref{theorem:diff}]
  Denote the differential of a function $g$ at $\mathcal{X}$ with
  increment $\Delta$ as $dg[\Delta]$. Since $U_k$ and $D_k$ are the
  left singular vectors and the singular values, respectively, of
  $\mathcal{X}_{(k)}$ for each $k = 1,\ldots,K$, the differentials,
  $dU_k[\Delta]$ and $dD_k[\Delta]$, are the same as in
  \cite{candes2013unbiased} and have a closed form solution, given by
  \begin{align}
    \label{equation:d.D}
    d\sigma_i^k[\Delta] = (U_k^T\Delta_{(k)}U_{-k}\mathcal{S}_{(k)}D_k^{-1})_{[i,i]} \text{ for } i =1,\ldots,p_k \text{ and } k=1,\ldots,K,
  \end{align}
  where
  \begin{align*}
    U_{-k} = U_K\otimes \cdots \otimes U_{k+1} \otimes U_{k-1}\otimes\cdots\otimes U_1.
  \end{align*}
  This follows because the SVD of $\mathcal{X}_{(k)}$ is $U_kD_kV_k^T
  = U_k\mathcal{S}_{(k)}U_{-k}^T$ which implies that $V_k =
  U_{-k}\mathcal{S}_{(k)}^TD_{k}^{-1}$. We plug in $V_k$ into equation
  (4.7) of \cite{candes2013unbiased} to get (\ref{equation:d.D}).

  Let $\Omega_{U_k}[\Delta] = U_k^TdU_k[\Delta]$. Then from (4.8) of
  \cite{candes2013unbiased} we have
  \begin{align}
    \begin{split}
      \label{equation:d.omega}
      &\Omega_{U_k}[\Delta]_{[i,j]} \\
      &= - 1(i \neq j)\left[\sigma_j^k(U_k^T\Delta_{(k)}U_{-k}S_{(k)}^TD_k^{-1})_{[i,j]} + \sigma_i^k(U_k^T\Delta_{(k)}U_{-k}S_{(k)}^TD_k^{-1})_{[j,i]}\right]/((\sigma_i^k)^2 - (\sigma_j^k)^2),
    \end{split}
  \end{align}
  and so
  \begin{align}
    \label{equation:d.U}
    dU_k[\Delta] = U\Omega_{U_k}[\Delta].
  \end{align}

  We now derive $d\mathcal{V}[\Delta]$. Let $U = (U_1,\ldots,U_K)$ and
  $D = (D_1,\ldots,D_K)$. Also note that $d\mathcal{X}[\Delta] =
  \Delta$. Using the chain rule, and following Chapter 8, Section 1,
  Equations (15) and (16) of \cite{magnus1999matrix} for the
  differential of matrix multiplication and the Kronecker product, we
  have
  \begin{align}
    \Delta = d\mathcal{X}[\Delta] &= d(U\cdot D \cdot\mathcal{V})[\Delta] \nonumber\\
    &=\sum_{k=1}^Kd\underline{U}_k[\Delta] \cdot D \cdot \mathcal{V} + \sum_{k=1}^K U \cdot d\underline{D}_k[\Delta] \cdot \mathcal{V} + U \cdot D \cdot d\mathcal{V}[\Delta], \label{equation:chain.rule}
  \end{align}
  where
  \begin{align}
    &d\underline{U}_k[\Delta] = (U_1,\ldots,U_{k-1},dU_k[\Delta],U_{k+1},\ldots,U_K) \text{ and} \label{equation:U.underline}\\
    &d\underline{D}_k[\Delta] = (D_1,\ldots,D_{k-1},dD_k[\Delta],D_{k+1},\ldots,D_K).
  \end{align}

  From (\ref{equation:chain.rule}), we solve for
  $d\mathcal{V}[\Delta]$ and have
  \begin{align}
    \label{equation:d.V}
    d\mathcal{V}[\Delta] = D^{-1} \cdot U^T \cdot \Delta - \sum_{k=1}^K dF_k[\Delta]\cdot \mathcal{V} - \sum_{k=1}^K dG_k[\Delta] \cdot \mathcal{V},
  \end{align}
  where
  \begin{align}
    \label{equation:d.FG}
    &dF_k[\Delta] = (I_{p_1},\ldots,I_{p_{k-1}},D_k^{-1}\Omega_{U_k}[\Delta]D_k,I_{p_{k+1}},\ldots,I_{p_K}) \text{ and}\\
    &dG_k[\Delta] = (I_{p_1},\ldots,I_{p_{k-1}},D_k^{-1}dD_k[\Delta],I_{p_{k+1}},\ldots,I_{p_K}).
  \end{align}
\end{proof}

\subsection{Divergence of higher-order spectral estimators}
\label{subsection:div}
In this section, we show that the divergence of higher-order spectral
estimators of the form (\ref{equation:ho.spect.est}) can be found in
the following theorem.
\begin{theorem}
  The divergence of estimators of the form
  (\ref{equation:ho.spect.est}) is
\begin{align}
  \sumo\left(f(D) \cdot D^{-1}\cdot \mathcal{C} + \sum_{k=1}^KH_k \cdot \mathcal{S}^2\right), \label{equation:divergence.final}
\end{align}
where $\sumo(\mathcal{A})$ is the sum of all elements in the tensor
$\mathcal{A}$, $\mathcal{S}^2 \in \mathbb{R}^{p_1\times\cdots\times
  p_K}$ such that $(\mathcal{S}^2)_{[\mathbf{i}]} =
(\mathcal{S}_{[\mathbf{i}]})^2$,
\begin{align}
  \label{equation:H.k}
  H_k = (f^1(D_1)D_1^{-1},\ldots,f^{k-1}(D_{k-1})D_{k-1}^{-1},D_{k}^{-1}df^k(D_k)D_k^{-1},f^{k+1}(D_{k+1}),\ldots,f^K(D_K)),
\end{align}
and $\mathcal{C} \in \mathbb{R}^{p_1\times\cdots\times p_K}$ such that
\begin{align}
  \label{equation:C.array}
  \mathcal{C}_{[{\bf i}]} =  1 + \sum_{k=1}^K \sum_{j = 1, j\neq i_k}^{p_k} \frac{\mathcal{S}_{[i_1,\ldots,i_{k-1},j,i_{k+1},\ldots,i_K]}^2}{(\sigma_{i_k}^k)^2 - (\sigma_{j}^k)^2} - \mathcal{S}_{\bf [i]}^2\sum_{k=1}^K\left( \frac{1}{(\sigma_{i_k}^k)^2} + \sum_{m=1,m\neq i_k}^{p_k}\frac{1}{(\sigma_m^k)^2 - (\sigma_{i_k}^k)^2}  \right).
\end{align}
\end{theorem}

\begin{proof}
  Let
  \begin{align*}
    \Delta^{i_1,\ldots,i_K} = \Delta^{\bf i} = U_{1[:,i_1]} \circ \cdots \circ U_{K[:,i_K]},
  \end{align*}
  where $\circ$ is the outer product and $U_{k[:,i_k]}$ is the $i_k$th
  column of $U_k$. Note that
  \begin{align*}
    (U_1^T,\ldots,U_K^T)\cdot\Delta^{\bf i} = E^{\bf i},
  \end{align*}
  where $E^{\bf i}$ is the $p_1\times\cdots\times p_K$ array with a
  one in position $(i_1,\ldots,i_K)$ and zeros everywhere
  else. Similar to the arguments of \cite{candes2013unbiased}, also
  note that $\Delta^{\bf i}$ forms an orthonormal basis for
  $\mathbb{R}^{p_1\times\cdots\times p_K}$, and so
  \begin{align}
    \diverge(t(\mathcal{X})) &= \sum_{\bf i} \langle\Delta^{\bf i}, df[\Delta^{\bf i}]\rangle \nonumber\\
    &= \sum_{\bf i} \langle(U_1^T,\ldots,U_K^T)\cdot\Delta^{\bf i}, (U_1^T,\ldots,U_K^T)\cdot df[\Delta^{\bf i}]\rangle\nonumber\\
    &= \sum_{\bf i} \langle E^{\bf i}, (U_1^T,\ldots,U_K^T)\cdot df[\Delta^{\bf i}]\rangle,\nonumber\\
    \label{equation:divergence.index}
    &= \sum_{\bf i} ((U_1^T,\ldots,U_K^T)\cdot df[\Delta^{\bf i}])_{[\mathbf{i}]},
  \end{align}
  where $\langle , \rangle$ is the usual Euclidean inner product. From
  the chain rule, we have:
  \begin{align*}
    df[\Delta^{\bf i}] = \sum_{k=1}^Kd\underline{U}_k[\Delta^{\bf i}] \cdot f(D) \cdot \mathcal{V} + \sum_{k=1}^K U \cdot df(\tilde{D})_k[\Delta^{\bf i}] \cdot \mathcal{V} + U \cdot f(D) \cdot d\mathcal{V}[\Delta^{\bf i}],
  \end{align*}
  where
  \begin{align*}
    &f(D) = (f^1(D_1),\ldots,f^K(D_K)) \text{ and}\\
    &df(\tilde{D})_k[\Delta^{\bf i}] = (f^1(D_1),\ldots,f^{k-1}(D_{k-1}),d(f^k \circ D_k)[\Delta^{\bf i}],f^{k+1}(D_{k+1}),\ldots,f^K(D_K)),
  \end{align*}
  where ``$\circ$'' now means composition. Hence,
  \begin{align}
    \label{equation:u.times.df}
    U^T \cdot df[\Delta^{\bf i}] = \sum_{k=1}^Kd\tilde{U}_k[\Delta^{\bf i}] \cdot f(D) \cdot \mathcal{V} + \sum_{k=1}^K df(\tilde{D})_k[\Delta^{\bf i}] \cdot \mathcal{V} + f(D) \cdot d\mathcal{V}[\Delta^{\bf i}],
  \end{align}
  where
  \begin{align}
    d\tilde{U}_k[\Delta^{\bf i}] = (I_{p_1},\ldots,I_{p_{k-1}},\Omega_{U_k}[\Delta^{\bf i}],I_{p_{k+1}},\ldots,I_{p_K}). \label{equation:U.tilde}
  \end{align}

  The outline of the derivation of the divergence is as follows. The
  ultimate goal is to obtain the $(i_1,\ldots,i_K)$th element of $U^T
  \cdot df[\Delta^{\bf i}]$ in (\ref{equation:u.times.df}) and plug
  that into (\ref{equation:divergence.index}). We will first calculate
  all of the differentials that are in (\ref{equation:u.times.df}),
  then we will determine the $(i_1,\ldots,i_K)$th element of $U^T
  \cdot df[\Delta^{\bf i}]$. Then we will simplify
  (\ref{equation:divergence.index}). These latter two steps may be
  found in Appendix \ref{sec:simp.div}.

  We begin with the differentials. From (\ref{equation:d.D}), we have
  \begin{align}
    d\sigma_j^k[\Delta^{\bf i}] &= (U_k^T \Delta_{(k)}^{\bf i} U_{-k} S_{(k)}^T D_k^{-1})_{[j,j]} \nonumber\\
    &= (E^{\bf i}_{(k)} S_{(k)}^T D_k^{-1})_{[j,j]}\nonumber\\
    \label{equation:d.sigma.i}
    &= 1(j = i_k)S_{[i_1,\ldots,i_{k-1},j,i_{k+1},\ldots,i_K]}/\sigma_j^k.
  \end{align}
  This is since $E_{(k)}^{\bf i}S_{(k)}^T \in \mathbb{R}^{p_k\times p_k}$ such that
  \begin{align}
    \left(E_{(k)}^{\bf i}S_{(k)}^T\right)_{[\ell,j]} = \label{equation:Ek.Sk}
    \begin{cases}
      0 &\text{if } \ell \neq i_k\\
      S_{[i_1,\ldots,i_{k-1},j,i_{k+1},\ldots,i_K]} &\text{if } \ell = i_k.
    \end{cases}
  \end{align}
  Similarly, from (\ref{equation:d.omega}), we have
  \begin{align}
    &\Omega_{U_k}[\Delta^{\bf i}]_{[\ell, j]} \nonumber\\
    &= -1(\ell\neq j)\left[\sigma_j^k(U_k^T\Delta_{(k)}U_{-k}S_{(k)}^TD_k^{-1})_{[\ell,j]} + \sigma_{\ell}^k(U_k^T\Delta_{(k)}U_{-k}S_{(k)}^TD_k^{-1})_{[j,\ell]}\right]/((\sigma_{\ell}^k)^2 - (\sigma_j^k)^2)\nonumber\\
    &= -1(\ell\neq j)\left[\sigma_j^k(E^{\bf i}_{(k)}S_{(k)}^TD_k^{-1})_{[\ell,j]} + \sigma_{\ell}^k(E^{\bf i}_{(k)}S_{(k)}^TD_k^{-1})_{[j,\ell]}\right]/ ((\sigma_{\ell}^k)^2 - (\sigma_j^k)^2)\nonumber\\
    \label{equation:d.omega.i}
    &= -1(\ell\neq j)\left[S_{[i_1,\ldots,i_{k-1},j,i_{k+1},\ldots,i_K]}1(\ell = i_k) + S_{[i_1,\ldots,i_{k-1},\ell,i_{k+1},\ldots,i_K]}1(j=i_k)\right]/((\sigma_{\ell}^k)^2 - (\sigma_j^k)^2).
  \end{align}
  Also, from the chain rule, we have that
  \begin{align}
    d(f^k_j \circ \sigma^k_j)[\Delta^{\bf i}] &= \left(\frac{d}{d\sigma_j^k}f_j^k(\sigma_j^k)\right)d\sigma_j^k[\Delta^{\bf i}]\nonumber\\
    \label{equation:d.comp}
    &= \delta_{j,i_k}\left(\frac{d}{d\sigma_j^k}f_j^k(\sigma_j^k)\right)S_{[i_1,\ldots,i_{k-1},j,i_{k+1},\ldots,i_K]}/\sigma_j^k.
  \end{align}

  We have just completed all of the calculus necessary to obtain the
  divergence, and the remainder of the calculation is
  simplification. That is, we can use equations (\ref{equation:d.V}),
  (\ref{equation:divergence.index}), (\ref{equation:u.times.df}),
  (\ref{equation:d.sigma.i}), (\ref{equation:d.omega.i}), and
  (\ref{equation:d.comp}) to calculate a closed-form expression for
  the divergence. This simplification is relegated to Appendix
  \ref{sec:simp.div}.
\end{proof}

% Note that this form of the divergence is beneficial for numerical stability reasons because we can calculate first
% $f^k(D_k)D_k^{-1}$ for $k = 1,\ldots,K$, then Tucker product this list
% of matrices with an array that is a sum of elements that should not get
% too small (since they only contain ratios of $\mathcal{S}^2$ and the
% mode-specific singular values).

We now present the formula for the SURE for all higher-order spectral
estimators of the form (\ref{equation:ho.spect.est}):
\begin{theorem}[SURE for (\ref{equation:ho.spect.est})]
  Under the model (\ref{equation:normal.model}), suppose $t(\cdot)$ in
  (\ref{equation:ho.spect.est}) is almost differentiable and for which
  (\ref{equation:int.condition}) holds. Then
  \begin{align}
    \sure(t) =  ||t(\mathcal{X}) - \mathcal{X}||^2 + 2\tau^2\sumo\left(f(D) \cdot D^{-1}\cdot \mathcal{C} + \sum_{k=1}^KH_k \cdot \mathcal{S}^2\right) - p\tau^2. \label{equation:complete.sure}
  \end{align}
\end{theorem}

This SURE formula is applicable for all shrinkage functions of the
form (\ref{equation:ho.spect.est}) where $f^k(D_k) =
\diag(f_1^k(\sigma_1^k),\ldots,f_{p_k}^k(\sigma_{p_k}^k))$. For such
shrinkage functions, the shrinkage being applied to each singular
value is a function only of that singular value. However, it is
possible to construct estimators which use all of the mode $k$
singular values to shrink each mode $k$ singular value, e.g.\ if we
were to use a shrinkage function analogous to those of
(\ref{equation:stein}) or (\ref{equation:improved.em}). For such
estimators, we prove in Appendix \ref{sec:gen.spec.func} that the form
of the divergence is very similar as in
(\ref{equation:divergence.final}). The only difference is that one
replaces $\frac{d}{d\sigma_{i_k}^k}f_{i_k}^k(\sigma_{i_k}^k)$ with
$\frac{d}{d\sigma_{i_k}^k}f_{i_k}^k(\sigma_{1}^k,\ldots,\sigma_{p_k}^k)$. That
is, for such shrinkage functions, $df^k(D_k)$ is a diagonal matrix
containing only the diagonal of the Jacobian matrix of the
transformation $\diag(D_k) \mapsto \diag(f(D_k))$.

\section{Simulation studies}
\label{sec:simulation}

In this section, we consider four competitors to the mode-specific
soft-thresholding estimator (\ref{equation:msst.est}) and the
truncated HOSVD (\ref{equation:trunc.shrink}). We will compare these
estimators assuming the error variance $\tau^2$ is one. The first
competitor is $\mathcal{X}$, which is the maximum likelihood estimator
and the uniformly minimum variance unbiased estimator. However, the
risk-performance of this estimator is known to be dominated by our
second competitor, the James-Stein estimator (\ref{equation:stein})
\citep{stein1981estimation}.  This estimator may be derived from an
empirical Bayes argument where $\Theta_{[\mathbf{i}]} \sim
N(0,\gamma^2)$ \citep{efron1972limiting}. As such, it should perform
well when the entries of $\Theta$ are centered about $0$. For a matrix
parameter $\Theta$, \cite{efron1972empirical} developed an empirical
Bayes estimator that performs better than the James-Stein estimator
when $\Theta$ exhibits empirical correlation along the rows. With this
in mind, our third estimator is obtained by applying the Efron-Morris
estimator (\ref{equation:efron.morris}) to the first mode
matricization of the data tensor. However, the Efron-Morris estimator
does not induce low rank estimates, and so our fourth and final
competitor is the matrix soft-thresholding estimator
(\ref{equation:soft.thresholding}) applied to the first mode
matricization of $\mathcal{X}$, and whose tuning parameter is chosen
with the SURE formula from \cite{candes2013unbiased}. This estimator
should improve on the Efron-Morris estimator when $\Theta_{(1)}$ has
approximately low rank.

We now describe the design of the simulation study. We evaluated the
risk of the mode-specific soft-thresholding, truncated HOSVD, maximum
likelihood, James-Stein, Efron-Morris, and matrix soft-thresholding
estimators under six different values of $\Theta \in \mathbb{R}^{10
  \times 10 \times 10}$, constructed as follows:
\begin{description}[noitemsep]
\item[A.] $\veco(\Theta) \sim N_p(0,I_{1000})$.
\item[B.] $\veco(\Theta) \sim N_p(0,I_{10} \otimes I_{10} \otimes F)$, where $F =
  \diag(1^2,2^2,\ldots,10^2)$.
\item[C.] $\veco(\Theta) \sim N_{1000}(0,I_{10}\otimes I_{10}\otimes \Sigma)$ where
  $\Sigma \in \mathbb{R}^{10\times 10}$ has an AR-1 $(0.7)$ covariance structure. That is,
  $\Sigma_{[i,j]} = 0.7^{|i - j|}$.
\item[D.] $\Theta_{(1)} = U_{[:,1:5]}D_{[1:5,1:5]}V_{[:,1:5]}^T$ where
  $UDV^T$ is the SVD of a $10 \times 100$ matrix that has standard
  normal entries.
\item[E.] $\veco(\Theta) \sim N_p(0,F \otimes F \otimes F)$, where $F =
  \diag(1^2,2^2,\ldots,10^2)$.
\item[F.] $\Theta$ is a rank $(5,5,5)$ tensor where all of the non-zero mode-specific
  singular values are the same along all modes.
\end{description}
% To get such a $\Theta$, we first drew $\veco(\mathcal{Y}) \sim
% N_{125}(0,I_{125})$ where $\mathcal{Y} \in \mathbb{R}^{5\times
% 5\times 5}$. We then found the ISVD of $Y =
% \ell(U_1,U_2,U_3)\cdot(D_1,D_2,D_3)\cdot\mathcal{V}$ from
% \cite{gerard2014higher}, which is different from
% (\ref{equation:hosvd.rewrite}). We then Set $\Theta =
% (\tilde{U}_1,\tilde{U}_2,\tilde{U}_{3})\cdot\mathcal{V}$, where
% $\tilde{U}_k$ was drawn uniformly from the Stiefel manifold of $10
% \times 5$ matrices.
For each scenario, we re-scaled $\Theta$ to have Frobenius norm
$\sqrt{1000}$, so that $E[||\mathcal{E}||^2] = 1000 =
||\Theta||^2$. For each $\Theta$, we simulated
$\mathcal{X}_{[\mathbf{i}]} \sim N(\Theta_{[\mathbf{i}]},1)$,
calculated the six estimators given this data tensor, and calculated
the squared error loss for each estimator. We repeated this process
500 times. Box plots of the losses for each of the six $\Theta$ values
are given in Figure \ref{fig:sim.results}.

The James-Stein estimator (\ref{equation:stein}) is expected to
perform well in Scenario \textbf{A} as it can be viewed as an
empirical Bayes procedure for the prior with which $\Theta$ was
actually generated.
% the mean was sampled under the conditions that the empirical Bayes
% estimator was derived.
Indeed, from Figure \ref{fig:sim.results} (\textbf{A}), the
James-Stein estimator does perform best, but the mode-specific
soft-thresholding estimator performs almost as well, even though there
is no correlation along any of the modes of the mean tensor.

For scenario \textbf{B}, we expect the matrix soft-thresholding
estimator (\ref{equation:soft.thresholding}) to do well. Since the
mean tensor in this scenario has approximately low rank only along the
first mode, estimators that shrink towards the space of low
multilinear rank tensors should be over-fitting and should not perform
well. From Figure \ref{fig:sim.results} (\textbf{B}), the matrix
soft-thresholding estimator does perform best, but surprisingly, the
mode-specific soft-thresholding estimator does equally well.

For Scenario \textbf{C}, we expect the matrix soft-thresholding
estimator (\ref{equation:soft.thresholding}) and the Efron-Morris
estimator (\ref{equation:efron.morris}) to perform well. There is
temporal correlation along one of the modes of the mean tensor. We
take into account the temporal correlation of the mean by performing
soft-thresholding along this mode. However, from Figure
\ref{fig:sim.results} (\textbf{C}), we see that the mode-specific
soft-thresholding estimator performed best.

The matrix soft-thresholding estimator
(\ref{equation:soft.thresholding}) was designed to do well when the
mean matrix is of low rank. This is exactly the situation in Scenario
\textbf{D}, as a tensor with low rank along one mode may be matricized
to form a low rank matrix. However, from Figure \ref{fig:sim.results}
(\textbf{D}), for our one $\Theta$ value, the mode-specific
soft-thresholding estimator performs best.

As for Scenario \textbf{E}, we expect the mode-specific
soft-thresholding estimator (\ref{equation:msst.est}) to do well, as
the mean tensor has approximately low multilinear rank, but it is not
exactly low multilinear rank. Figure \ref{fig:sim.results}
(\textbf{E}) reveals the mode-specific soft-thresholding estimator
does indeed perform better than the other estimators.

We expect the truncated HOSVD (\ref{equation:trunc.shrink}) to do well
in Scenario \textbf{F} because the mean tensor has low multilinear
rank, and the truncated HOSVD is correctly shrinking toward this
structure. From Figure \ref{fig:sim.results} (\textbf{F}), we see that
the truncated HOSVD does indeed perform best in terms of
loss. However, the mode-specific soft-thresholding estimator does not
perform much worse. The estimators that do not take into account the
tensor indexing perform about twice as bad as these tensor-specific
estimators.

For scenarios \textbf{C} and \textbf{D}, we emphasize here that we are
looking at the risk only at a few points in the parameter space. There
are likely points where the matrix-soft thresholding estimator
performs better than the tensor estimators. However our mode-specific
soft-thresholding estimator did not perform poorly under any of our
simulated mean tensors.

\begin{figure*}
  \begin{center}
    \includegraphics{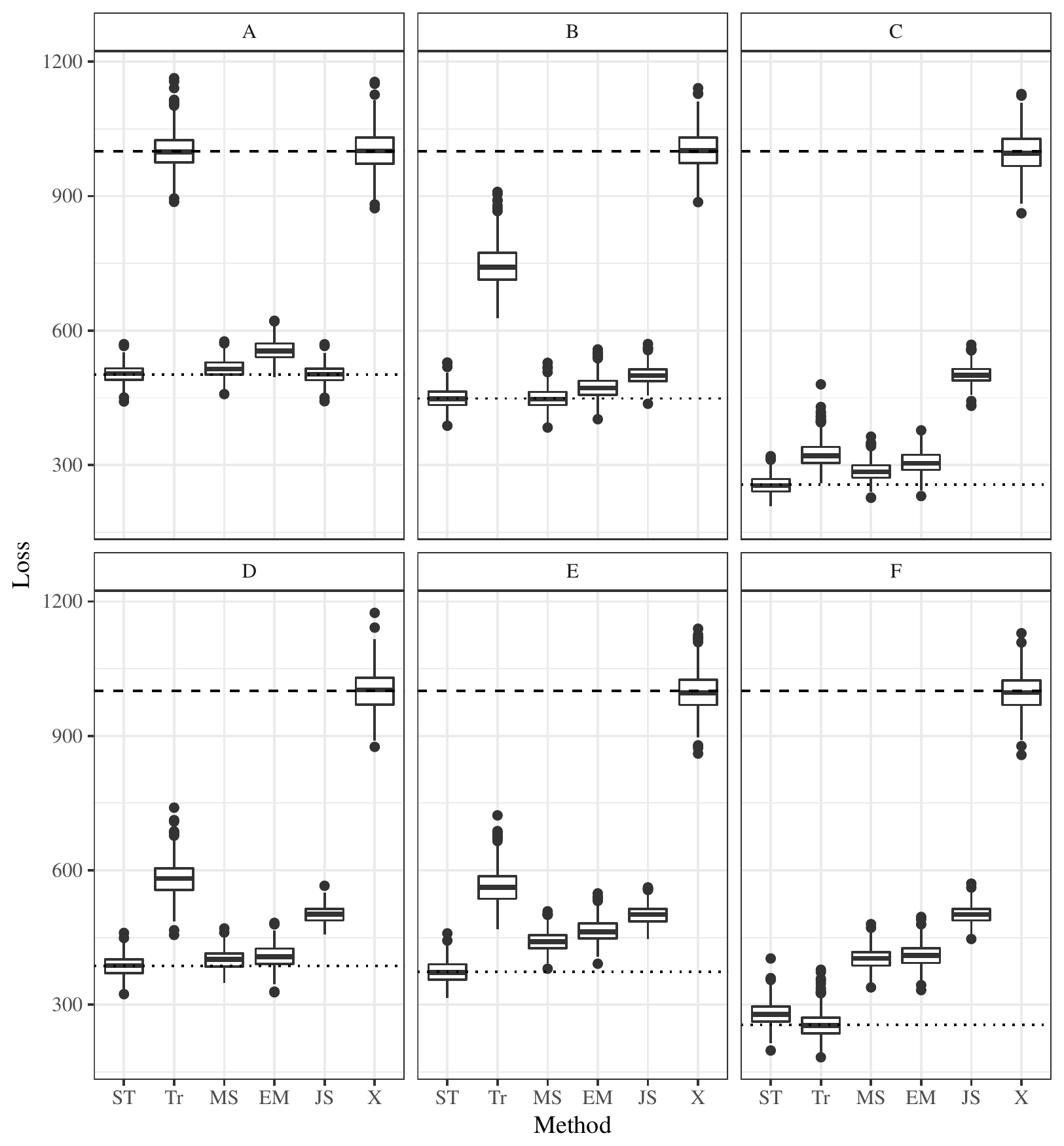}
    \caption{Box plots of losses for the six estimators under different
      scenarios. The estimators include the mode-specific
      soft-thresholding (ST), truncated HOSVD (Tr), matrix
      soft-thresholding (MS), Efron-Morris (EM), James-Stein (JS), and
      maximum likelihood (X) estimators. In the scenarios, the mean tensor
      was simulated to have (\textbf{A}) uncorrelated elements,
      (\textbf{B}) full rank but dispersed singular values only along mode
      1, (\textbf{C}) AR-1 covariance along mode 1, (\textbf{D}) low rank
      only along mode 1, (\textbf{E}) full rank but dispersed singular
      values along all modes, and (\textbf{F}) rank $(5,5,5)$ with all the
      same non-zero singular values.}
    \label{fig:sim.results}
  \end{center}
\end{figure*}

Our procedure for the truncated HOSVD produces a multilinear rank with
the smallest SURE. It is of interest to know if this multilinear rank
provides a good estimate of the true rank of $\Theta$. We evaluated
this possibility in simulation Scenarios \textbf{D} and \textbf{F}. In
Scenario \textbf{F}, where the tensor had dimension $(10,10,10)$ and
the true multilinear rank was $(5,5,5)$, this SURE method correctly
estimated the multilinear rank in 92.6\% of trials. In Scenario
\textbf{D}, where the true multilinear rank was $(5,10,10)$, the
results of the simulation study can be found in Table
\ref{tab:rank.est}. There, we see that the rank of the first mode is
correctly estimated in 97\% of trials. The rank of the second and
third modes are correctly estimated a majority of the time.

% latex table generated in R 3.3.2 by xtable 1.8-2 package
% Mon Feb 13 10:53:01 2017
\begin{table}[ht]
\centering
\begin{tabular}{rlllllll}
  \hline
 Estimated Rank & 4 & 5 & 6 & 7 & 8 & 9 & 10 \\
  \hline
Mode 1 & .03 & .97 & 0 & 0 & 0 & 0 & 0 \\
  Mode 2 & 0 & 0 & .02 & .03 & .11 & .27 & .57 \\
  Mode 3 & 0 & 0 & 0 & .01 & .05 & .18 & .74 \\
   \hline
\end{tabular}
\caption{Proportion of times each rank is estimated based on SURE for each mode over 500 repetitions when the true multilinear rank is (5, 10, 10).}
\label{tab:rank.est}
\end{table}

\section{Multivariate relational data example}
\label{sec:NBA}

In this section, we demonstrate the applicability of our estimators to
multivariate relational data.  Such data may be viewed as a three-way
tensor $\mathcal{X}$ where entry $\mathcal{X}_{[i,j,k]}$ is the value
of relation type $k$ from node $i$ to node $j$. One example of such a
data set is a social network in which multiple types of relations are
measured between individuals.  As another example, in sports
statistics, round robin interaction data consist of outcomes of
competitions between teams. In this section we illustrate our methods
with round robin data from the 2014-2015 regular season of the
National Basketball Association (NBA). The NBA consists of a Western
conference and an Eastern conference of fifteen teams each, where
intra-conference play has three to four games per year per pair of
teams and inter-conference play is limited to two games a season per
pair of teams. For each conference, we created a four dimensional
tensor where element $\mathcal{Y}_{[i,j,k,\ell]}$ is statistic $k$
obtained by team $i$ while playing team $j$ either during team $i$'s
first home ($\ell = 1$) or first away ($\ell = 2$) game against team
$j$ during the season. The statistics we considered were free-throw
percentage, two-point field goal percentage, and three-point field
goal percentage. We thus have two tensors each of dimension $15 \times
15 \times 3 \times 2$, one for each of the two conferences.  In this
section, we illustrate the utility of tensor shrinkage by predicting
late season relational basketball statistics from early season
data. Our approach is analogous to that of \cite{efron1975data}, who
illustrated the utility of vector shrinkage estimation by predicting
late season baseball batting averages from data on early season
batting averages.

The statistics in our data set are all empirical proportions. We model
the elements of $\mathcal{Y}$ with a binomial model,
\begin{align*}
  n_{i,j,k,\ell}\mathcal{Y}_{[i,j,k,\ell]} \sim \Bin(n_{i,j,k,\ell},p_{i,j,k,\ell}),
\end{align*}
where all elements are independent, given the $p_{i,j,k,\ell}$'s. We
apply an arc-sin transformation to the data tensor to stabilize the
variance:
\begin{align*}
  \mathcal{X}_{[i,j,k,\ell]} = (n_{i,j,k,\ell})^{1/2}\arcsin(2\mathcal{Y}_{[i,j,k,\ell]} - 1).
\end{align*}
From the central limit theorem, we have approximately
\begin{align*}
  \mathcal{X}_{[i,j,k,\ell]} \sim N(\Theta_{[i,j,k,\ell]},1),
\end{align*}
where $\Theta_{[i,j,k,\ell]} =
(n_{i,j,k,\ell})^{1/2}\arcsin(2p_{i,j,k,\ell} - 1)$, resulting in the
model in (\ref{equation:normal.model}).

A commonly used representation of a mean tensor $\Theta$ is an ANOVA
decomposition, such as
\begin{align*}
  \Theta_{[i,j,k,\ell]} = \mu + \alpha_i + \beta_j + \gamma_k + \delta_{\ell} + \tilde{\Theta}_{[i,j,k,\ell]},
\end{align*}
where $\tilde{\Theta}_{[i,j,k,\ell]}$ contains all of the interaction
effects. Note that $\mathbf{1}_{p_1}^T\alpha = 0$,
$\mathbf{1}_{p_2}^T\beta = 0$, $\mathbf{1}_{p_3}\gamma = 0$, and
$\mathbf{1}_{p_4}^T\delta = 0$, where $\mathbf{1}_{p_k}$ is the vector
of ones of length $p_k$. The tensor $\tilde{\Theta}$ also satisfies
$\Tilde{\Theta}_{(k)}\mathbf{1}_{p/p_k} = 0$ for all $k =
1,2,3,4$. Suppose we obtain the maximum likelihood estimates of $\mu$,
$\alpha$, $\beta$, $\gamma$, and $\delta$ by fitting a main-effects
ANOVA model. We then calculate the residual tensor,
\begin{align*}
%%\mathcal{R}_{[i,j,k,\ell]} = \mathcal{X}_{[i,j,k,\ell]} - \mathcal{X}_{[i\cdot\cdot\cdot]} - \mathcal{X}_{[\cdot j\cdot\cdot]} - \mathcal{X}_{[\cdot\cdot k\cdot]} - \mathcal{X}_{[\cdot\cdot\cdot\ell]} + 3\mathcal{X}_{[\cdot\cdot\cdot\cdot]},
  \mathcal{R}_{[i,j,k,\ell]} =& \mathcal{X}_{[i,j,k,\ell]} -
  \frac{p_1}{p}\sum_{j',k',\ell'}\mathcal{X}_{[i,j',k',\ell']} -
  \frac{p_2}{p}\sum_{i',k',\ell'}\mathcal{X}_{[i',j,k',\ell']} -
  \frac{p_3}{p}\sum_{i',j',\ell'}\mathcal{X}_{[i',j',k,\ell']} \\
  &- \frac{p_4}{p}\sum_{i',j',k'}\mathcal{X}_{[i',j',k',\ell]} + \frac{3}{p}\sum_{i',j',k',\ell'}\mathcal{X}_{[i',j',k',\ell']}.
\end{align*}
This residual tensor has an expected value of $\tilde{\Theta}$.  It
was proposed in \cite{stein1966approach} and \cite{efron1972empirical}
that we estimate the interaction effects $\tilde{\Theta}$ with a
vector shrinkage-type estimator on the residuals. If the interactions
$\tilde{\Theta}$ are close to zero --- when the interaction effects
are small --- then such estimators will adaptively shrink the
residuals towards zero. However, these estimators were developed to
adapt to patterns in vectors or matrices of residuals, and not tensors
of residuals. In contrast, our approach should be able to adapt to
these patterns along any of the four modes of the residual tensor.

We applied mode-specific soft-thresholding and the truncated HOSVD to
the array of residuals $\mathcal{R}$ from the main effects ANOVA
model. These methods suggest that the residual tensor should be
heavily shrunk both towards zero and towards low multilinear rank
structure. For the West, the Frobenius norm of the residual tensor was
38.38, while the Frobenius norm of the resulting shrunken residual
tensor using the mode-specific soft-thresholding estimator was
7.81. In the East, the values were 38.95 and 6.97, respectively. We
also used SURE to estimate the multilinear rank of each residual
tensor using the truncated HOSVD. The estimated multilinear rank of
the residual tensor of the Western conference was $2 \times 3 \times 1
\times 2$, and for the Eastern conference the estimated multilinear
rank was $4 \times 2 \times 1 \times 1$. These are very small ranks
compared to the dimensions of the tensors $15 \times 15 \times 3
\times 2$.

An ad hoc evaluation of the performance of our estimators can be
obtained by predicting game statistics after the first home and first
away games. Since some teams only play each other three times, we do
not have late season data on all possible combinations of team pairs
by home versus away games. For the late season data we do have, we
present the squared error losses for predicting the statistics of the
remaining part of the season for each conference in Table
\ref{tab:sel.nba}. The different estimators are (1) the raw data array
$\mathcal{X}$, (2) the mean estimates of the main-effects ANOVA model,
(3) the mode-specific soft-thresholding shrunken residual tensor added
to the mean estimates of the main-effects ANOVA model, (4) the
truncated HOSVD shrunken residual tensor added to the mean estimates
of the main-effects ANOVA model, and (5) an estimator derived from
logistic regression using the main-effects of each mode. The losses
are with respect to the arc-sin transformed data. The poor performance
of $\mathcal{X}$ is unsurprising. The amount of shrinkage that our
estimators produce indicates that the fully saturated model is
over-fitting and that most of the information is contained in the
main-effects. However, our mode-specific soft-thresholding estimator
is also fitting the fully saturated model and it performs comparable
to the main-effects ANOVA model, even improving the predictions for
the Eastern conference.

\begin{table}[ht]
  \centering
  \begin{tabular}{rrr}
    Estimator & East & West\\
    \hline
    $\mathcal{X}$ & 2410 & 2476\\
    ANOVA & 1344 & 1364 \\
    Mode-specific Soft-thresholding & 1327 & 1385 \\
    Truncated HOSVD & 1391 & 1451\\
    Logistic Regression & 1481 & 1552\\
  \end{tabular}
  \caption{Squared error losses when predicting the statistics of the remaining games of the season.}
  \label{tab:sel.nba}
\end{table}

\section{Discussion}
\label{sec:discussion}

This paper introduced new classes of shrinkage estimators for
tensor-valued data that are higher-order generalizations of existing
matrix spectral estimators. Each class is indexed by tuning parameters
whose values we chose by minimizing an unbiased estimate of the
risk. In terms of MSE, these estimators outperform their matrix
counterparts when the mean has approximately low multilinear rank and
they perform competitively when the mean does not have low multilinear
rank.

There has been some recent work on penalized optimization methods for
estimating signal tensors in the presence of Gaussian noise
\citep{signoretto2010convex,tomioka2011estimation,tomioka2011statistical,liu2013tensor,tomioka2013convex}. Usually,
these estimators are defined as the minimizers of a penalized squared
error empirical loss, where the penalty is usually some generalization
of the nuclear norm to tensors (for example, the sum of the nuclear
norms of the $K$ matricizations of a tensor). These estimators, though
similar in spirit, are very different from our approach. The main
advantage of our estimators is their simplicity --- they are simply
functions of the HOSVD (\ref{equation:ho.spect.est}) for which there
are efficient and accurate numerical procedures to compute.

We have presented a way to adaptively choose the tuning parameters of
our higher-order spectral estimators by minimizing the SURE. This
approach is applicable, not just for the truncated HOSVD
(\ref{equation:trunc.shrink}) and the mode-specific soft-thresholding
(\ref{equation:msst.est}) estimators, but also for \emph{all}
estimators of the form (\ref{equation:ho.spect.est}) that satisfy the
conditions of Theorem \ref{theorem:stein.theorem}.  Although we found
that adaptively choosing the tuning parameters by minimizing the SURE
worked well under the scenarios we studied, there are other ways to
select tuning parameters. In the case of matrix spectral estimators,
others have chosen the amount of shrinkage by minimax considerations
\citep{efron1972empirical,stein1981estimation}, cross-validation
\citep{bro2008cross,owen2009bi,josse2012selecting}, and asymptotic
considerations
\citep{gavish2014optimalhard,gavish2014optimal}. Exploring these
methods for our higher-order spectral estimators
(\ref{equation:ho.spect.est}) is a current research area of the
authors.

In this paper, we focused on estimators of the form
(\ref{equation:ho.spect.est}). If the mean tensor is believed to have
approximately low multilinear rank, we should shrink the core array
through the Tucker product along the modes to obtain this low
multilinear rank. The form of our higher-order spectral estimators
(\ref{equation:ho.spect.est}) allows us to use the mode-specific
singular values to determine the form and amount of shrinkage that
should be performed to each mode of the core array. However, different
classes of higher-order spectral estimators can be studied. In the
Appendix \ref{sec:sure.s.shrink}, we explore functions that shrink
each element of the core array individually:
\begin{align*}
t(\mathcal{X}) = (U_1,\ldots,U_K)\cdot g(\mathcal{S}), \text{ where } g(\mathcal{S})_{[\mathbf{i}]} = g_{\mathbf{i}}(\mathcal{S}_{[\mathbf{i}]}).
\end{align*}
This class of estimators can be used, for example, to induce zeros in
the core array, which has applications in increasing the
interpretability of a higher-order generalization of principal
components analysis
\citep{henrion1993body,kiers1997uniqueness,murakami1998case,andersson1999general,de2001independent,martin2008jacobi}.

Although the error variance $\tau^2$ in (\ref{equation:normal.model})
might be known in some settings, such as fMRI data sets
\citep{candes2013unbiased}, in most applied situations the variance
would not be unknown. There are matrix-specific estimates of the
variance that can be applied to tensor-variate datasets by first
matricizing along each mode. In our software, we have implemented the
methods described in \citet{choi2014selecting} and
\citet{gavish2014optimalhard}. Though, instead of plugging in an
estimate of the variance into the SURE formula
(\ref{equation:sure.form}), there has been a recent suggestion to use
a generalized SURE formula
\citep{sylvain2012smooth,josse2015adaptive}:
\begin{align*}
\gsure(t) = \frac{||t(\mathcal{X}) - \mathcal{X}||^2}{(1 - \diverge(t(\mathcal{X}))/p)^2}.
\end{align*}
This formula is motivated by generalized cross-validation
\citep{golub1979generalized} and is an approximation to SURE
\citep{josse2015adaptive}. Importantly, GSURE does not require the
variance to be known, and so its minimization may be accomplished
without an estimate of $\tau^2$. For our higher-order spectral
estimators, we have already accomplished the hard work of calculating
the divergence in this paper, and implementing GSURE is an easy
application of this result. Our software allows for GSURE
implementation for the estimators discussed in this article.

All methods discussed in this paper are implemented in the R package
\texttt{hose} available at
\begin{center}
  \href{https://github.com/dcgerard/hose}{https://github.com/dcgerard/hose}.
\end{center}
Code and instructions to reproduce all of the results of this paper are available at
\begin{center}
\href{https://github.com/dcgerard/hose\_paper/tree/master/reproduce\_sure}{https://github.com/dcgerard/hose\_paper/tree/master/reproduce\_sure}.
\end{center}
\appendix
\section{Simplification of the divergence}
\label{sec:simp.div}
We will need the $(i_1,\ldots,i_K)$th element of $U^T \cdot
df[\Delta^{\bf i}]$ in (\ref{equation:u.times.df}). There are three
terms in (\ref{equation:u.times.df}). We will deal with them one by
one. First, we will work with the first term of
(\ref{equation:u.times.df}), $\sum_{k=1}^Kd\tilde{U}_k[\Delta^{\bf i}]
\cdot f(D) \cdot \mathcal{V}$.  Note that, for $\mathcal{A} =
f(D)\cdot\mathcal{V}$, we have
\begin{align*}
  &\left(d\tilde{U}_k[\Delta^{\bf i}] \cdot \mathcal{A}\right)_{[\mathbf{i}]} = \left((I_{p_1},\ldots,I_{p_{k-1}},\Omega_{U_k}[\Delta^{\bf i}],I_{p_{k+1}},\ldots,I_{p_K})\cdot \mathcal{A}\right)_{[\mathbf{i}]}\\
  &= -\sum_{j=1, j\neq i_k}^{p_k}\mathcal{S}_{[i_1,\ldots,i_{k-1},j,i_{k+1},\ldots,i_K]}\mathcal{A}_{[i_1,\ldots,i_{k-1},j,i_{k+1},\ldots,i_K]}/[(\sigma_{i_k}^k)^2 - (\sigma_j^k)^2]\\
  &= -\sum_{j=1, j\neq i_k}^{p_k}\left(\prod_{\ell=1,\ell \neq k}^Kf_{i_{\ell}}^{\ell}(\sigma_{i_{\ell}}^{\ell})\right)f_j^k(\sigma_j^k)\mathcal{S}_{[i_1,\ldots,i_{k-1},j,i_{k+1},\ldots,i_K]}\mathcal{V}_{[i_1,\ldots,i_{k-1},j,i_{k+1},\ldots,i_K]}/[(\sigma_{i_k}^k)^2 - (\sigma_j^k)^2]\\
  &= -\left(\prod_{\ell=1,\ell \neq k}^Kf_{i_{\ell}}^{\ell}(\sigma_{i_{\ell}}^{\ell})\right)\sum_{j=1, j\neq i_k}^{p_k}f_j^k(\sigma_j^k)\mathcal{S}_{[i_1,\ldots,i_{k-1},j,i_{k+1},\ldots,i_K]}\mathcal{V}_{[i_1,\ldots,i_{k-1},j,i_{k+1},\ldots,i_K]}/[(\sigma_{i_k}^k)^2 - (\sigma_j^k)^2].
\end{align*}

Now we work with the second term of (\ref{equation:u.times.df}),
$\sum_{k=1}^K df(\tilde{D})_k[\Delta^{\bf i}] \cdot \mathcal{V}$.  We
have that:
\begin{align}
  \left(df(\tilde{D})_k[\Delta^{\bf i}] \cdot
    \mathcal{V}\right)_{[\mathbf{i}]} &= \left(\prod_{j\neq
      k}f_{i_j}^j(\sigma_{i_j}^j)\right)d(f_{i_k}^k \circ
  \sigma_{i_k}^k)[\Delta^{\bf i}]\mathcal{V}_{[\mathbf{i}]} \label{equation:simp.diff.first}\\
  &= \left(\prod_{j\neq k}f_{i_j}^j(\sigma_{i_j}^j)\right)\left(\frac{d}{d\sigma_{i_k}^k}f_{i_k}^k(\sigma_{i_k}^k)\right)\mathcal{V}_{[\mathbf{i}]}\mathcal{S}_{[\mathbf{i}]}/\sigma_{i_k}^k\nonumber\\
  &= \left(\prod_{j\neq k}f_{i_j}^j(\sigma_{i_j}^j)/\sigma_{i_j}^j\right)\left(\frac{d}{d\sigma_{i_k}^k}f_{i_k}^k(\sigma_{i_k}^k)\right)\mathcal{S}_{[\mathbf{i}]}^2/(\sigma_{i_k}^k)^2, \label{equation:final.form.second.term}
\end{align}
since $\mathcal{V}_{[\mathbf{i}]} =
\left(\prod_{k=1}^K\sigma_{i_k}^k\right)^{-1}\mathcal{S}_{[\mathbf{i}]}$.

It remains to work with the third term in (\ref{equation:u.times.df}),
$f(D) \cdot d\mathcal{V}[\Delta^{\bf i}]$.  We have:
\begin{align}
  \left(f(D) \cdot d\mathcal{V}[\Delta^{\bf i}]\right)_{[\mathbf{i}]} = \left(\prod_{k=1}^Kf_{i_k}^k(\sigma_{i_k}^k)\right) d\mathcal{V}[\Delta^{\bf i}]_{[\mathbf{i}]}. \label{eqaution:element.of.f.dv}
\end{align}
We now need to obtain $d\mathcal{V}[\Delta^{\bf i}]_{[\mathbf{i}]}$. From (\ref{equation:d.V}), we have
\begin{align}
  d\mathcal{V}[\Delta^{\bf i}] &= D^{-1} \cdot U^T \cdot \Delta^{\bf i} - \sum_{k=1}^K dF_k[\Delta^{\bf i}]\cdot \mathcal{V} - \sum_{k=1}^K dG_k[\Delta^{\bf i}] \cdot \mathcal{V},\nonumber\\
  &= D^{-1} \cdot E^{\bf i} - \sum_{k=1}^K dF_k[\Delta^{\bf i}]\cdot \mathcal{V} - \sum_{k=1}^K dG_k[\Delta^{\bf i}] \cdot \mathcal{V}.\label{equation:expand.d.v}
\end{align}
There are three terms in (\ref{equation:expand.d.v}). Let us deal with them one by one. The first term in (\ref{equation:expand.d.v}) is
\begin{align}
  \left(D^{-1} \cdot E^{\bf i}\right)_{[\mathbf{i}]} = \left(\prod_{k=1}^K\sigma_{i_k}^k\right)^{-1}.\label{equation:dv.1}
\end{align}
The second term in (\ref{equation:expand.d.v}) is
\begin{align}
  &\left( dF_k[\Delta^{\bf i}]\cdot \mathcal{V}\right)_{[\mathbf{i}]} \nonumber\\
  &= \left((I_{p_1},\ldots,I_{p_{k-1}},D_k^{-1}\Omega_{U_k}[\Delta^{\bf i}]D_k,I_{p_{k+1}},\ldots,I_{p_K})\cdot \mathcal{V}\right)_{[\mathbf{i}]}\nonumber\\
  &= \sum_{j=1}^{p_k}\left(D_k^{-1}\Omega_{U_k}[\Delta^{\bf i}]D_k\right)_{[i_k,j]} \mathcal{V}_{[i_1,\ldots,i_{k-1},j,i_{k+1},\ldots,i_K]}\nonumber\\
  &= -\sum_{j=1, j\neq i_k}^{p_k}\frac{\sigma_j^k}{\sigma_{i_k}^k}S_{[i_1,\ldots,i_{k-1},j,i_{k+1},\ldots,i_K]}\mathcal{V}_{[i_1,\ldots,i_{k-1},j,i_{k+1},\ldots,i_K]}/[(\sigma_{i_k}^k)^2 - (\sigma_j^k)^2]\nonumber\\
  &= -\sum_{j=1, j\neq i_k}^{p_k}\frac{\sigma_j^k}{\sigma_{i_k}^k}S_{[i_1,\ldots,i_{k-1},j,i_{k+1},\ldots,i_K]}\mathcal{V}_{[i_1,\ldots,i_{k-1},j,i_{k+1},\ldots,i_K]}/[(\sigma_{i_k}^k)^2 - (\sigma_j^k)^2]. \label{equation:dv.2}
\end{align}
The third term in (\ref{equation:expand.d.v}) is
\begin{align}
  \left(dG_k[\Delta^{\bf i}] \cdot \mathcal{V}\right)_{[\mathbf{i}]} &= \left(\mathcal{V} \times_k D_k^{-1}dD_k[\Delta^{\bf i}]\right)_{[\mathbf{i}]}\nonumber\\
  &= d\sigma_{i_k}^k[\Delta] \mathcal{V}_{[\mathbf{i}]}/\sigma_{i_k}^k\nonumber\\
  &= \mathcal{S}_{[\mathbf{i}]}\mathcal{V}_{[\mathbf{i}]}/(\sigma_{i_k}^k)^2.\label{equation:dv.3}
\end{align}
To obtain the third term in (\ref{equation:u.times.df}), we need only
plug in (\ref{equation:dv.1}), (\ref{equation:dv.2}), and
(\ref{equation:dv.3}) into (\ref{equation:expand.d.v}). And then we
need to plug in (\ref{equation:expand.d.v}) into
(\ref{eqaution:element.of.f.dv}).

We will now show that the divergence is of the form:
\begin{align}
  &\sum_{i_1,\ldots,i_K} \left[\mathcal{C}_{[\mathbf{i}]}\prod_{k=1}^Kf_{i_k}^k(\sigma_{i_k}^k)/\sigma_{i_k}^k + \sum_{k=1}^K \left(\prod_{j\neq k}f_{i_j}^j(\sigma_{i_j}^j)/\sigma_{i_j}^j\right)\left(\frac{d}{d\sigma_{i_k}^k}f_{i_k}^k(\sigma_{i_k}^k)\right)\mathcal{S}_{[i_1,\ldots,i_k]}^2/(\sigma_{i_k}^k)^2\right]\nonumber\\
  &= \sumo\left(f(D) \cdot D^{-1}\cdot \mathcal{C} + \sum_{k=1}^KH_k \cdot \mathcal{S}^2\right), \nonumber
\end{align}
for $H_k$ in (\ref{equation:H.k}) and $\mathcal{C} \in
\mathbb{R}^{p_1\times\cdots\times p_K}$ in
(\ref{equation:C.array}). The term $f(D) \cdot D^{-1}\cdot
\mathcal{C}$ is from the first and second parts of
(\ref{equation:u.times.df}), whereas the terms $\sum_{k=1}^KH_k \cdot
\mathcal{S}^2$ are from the second part of (\ref{equation:u.times.df})
and were already derived in
(\ref{equation:final.form.second.term}). Let us find $\mathcal{C}$.
Let $\mathbf{f}_{i_1,\ldots,i_k} = \mathbf{f}_{\mathbf{i}} =
\prod_{k=1}^Kf_{i_k}^k(\sigma_{i_k}^k)$. Ignoring the second term in
(\ref{equation:u.times.df}), we have that the sum of the first and
third terms in (\ref{equation:u.times.df}) is equal to:
\begin{align*}
  \begin{split}
    &\sum_{\bf i} \left\{ -\sum_{k=1}^K  \sum_{m=1,m\neq i_k}^{p_k} \mathbf{f}_{i_1,\ldots,i_{k-1},m,i_{k+1},\ldots,i_K}\frac{\mathcal{S}_{[i_1,\ldots,i_{k-1},m,i_{k+1},\ldots,i_K]}\mathcal{V}_{[i_1,\ldots,i_{k-1},m,i_{k+1},\ldots,i_K]}}{(\sigma_{i_k}^k)^2 - (\sigma_{m}^k)^2} \right.\\
    &\left. +\ \mathbf{f}_{\bf i}\left[ \left(\prod_{k=1}^K \sigma_{i_k}^k\right)^{-1} + \sum_{k=1}^K \sum_{j = 1, j\neq i_k}^{p_k} \frac{\sigma_j^k}{\sigma_{i_k}^k} \frac{\mathcal{S}_{[i_1,\ldots,i_{k-1},j,i_{k+1},\ldots,i_K]}\mathcal{V}_{[i_1,\ldots,i_{k-1},j,i_{k+1},\ldots,i_K]}}{(\sigma_{i_k}^k)^2 - (\sigma_{j}^k)^2}\right.\right.\\
    &\left.\left.- \mathcal{S}_{\bf [i]} \mathcal{V}_{\bf [i]}\sum_{k=1}^K \frac{1}{(\sigma_{i_k}^k)^2}\right]\right\}.
  \end{split}
\end{align*}
After rearranging summands, we obtain:
\begin{align*}
  \begin{split}
    &\sum_{\bf i} \mathbf{f}_{\bf i}\left[ \left(\prod_{k=1}^K \sigma_{i_k}^k\right)^{-1} + \sum_{k=1}^K \sum_{j = 1, j\neq i_k}^{p_k} \frac{\sigma_j^k}{\sigma_{i_k}^k} \frac{\mathcal{S}_{[i_1,\ldots,i_{k-1},j,i_{k+1},\ldots,i_K]}\mathcal{V}_{[i_1,\ldots,i_{k-1},j,i_{k+1},\ldots,i_K]}}{(\sigma_{i_k}^k)^2 - (\sigma_{j}^k)^2} \right.\\
    &\left. - \mathcal{S}_{\bf [i]} \mathcal{V}_{\bf [i]}\sum_{k=1}^K\left( \frac{1}{(\sigma_{i_k}^k)^2} + \sum_{m=1,m\neq i_k}^{p_k}\frac{1}{(\sigma_m^k)^2 - (\sigma_{i_k}^k)^2}  \right)\right].
  \end{split}
\end{align*}
And after factoring out $\prod_{k=1}^K(\sigma_{i_k}^k)^{-1}$, we get:
\begin{align*}
  \begin{split}
    &\sum_{\bf i} \mathbf{f}_{\bf i}\left(\prod_{k=1}^K \sigma_{i_k}^k\right)^{-1}\left[ 1 + \sum_{k=1}^K \sum_{j = 1, j\neq i_k}^{p_k} \frac{\mathcal{S}_{[i_1,\ldots,i_{k-1},j,i_{k+1},\ldots,i_K]}^2}{(\sigma_{i_k}^k)^2 - (\sigma_{j}^k)^2} \right.\\
    &\left.-\mathcal{S}_{\bf [i]}^2\sum_{k=1}^K\left( \frac{1}{(\sigma_{i_k}^k)^2} + \sum_{m=1,m\neq i_k}^{p_k}\frac{1}{(\sigma_m^k)^2 - (\sigma_{i_k}^k)^2}  \right)\right].
  \end{split}
\end{align*}
That is,
\begin{align}
  \label{equation:better.c}
  \mathcal{C}_{[{\bf i}]} =  1 + \sum_{k=1}^K \sum_{j = 1, j\neq i_k}^{p_k} \frac{\mathcal{S}_{[i_1,\ldots,i_{k-1},j,i_{k+1},\ldots,i_K]}^2}{(\sigma_{i_k}^k)^2 - (\sigma_{j}^k)^2} - \mathcal{S}_{\bf [i]}^2\sum_{k=1}^K\left( \frac{1}{(\sigma_{i_k}^k)^2} + \sum_{m=1,m\neq i_k}^{p_k}\frac{1}{(\sigma_m^k)^2 - (\sigma_{i_k}^k)^2}  \right).
\end{align}

\section{Details of optimization}
We now provide some brief details on our optimization strategy when considering only the mode-specific soft-thresholding estimator. Let
$f_{\mathbf{i}} = \prod_{k=1}^Kf_{i_k}^k(\sigma_{i_k}^k)$ and
$\tilde{\sigma}_{\mathbf{i}} = \prod_{k=1}^K\sigma_{i_k}^k$. The SURE
is equal to:
\begin{align}
  &||f(D) \cdot D ^{-1} \cdot \mathcal{S} - \mathcal{S}||^2 + 2\tau^2 \sum_{\mathbf{i}} \left[\left(f(D) \cdot D^{-1} \cdot \mathcal{C}\right)_{[\mathbf{i}]} + \sum_{k=1}^K \left(H_k \cdot \mathcal{S}^2\right)_{[\mathbf{i}]}\right] - p \tau^2\\
  & = \sum_{\mathbf{i}}\left[\left(f_{\mathbf{i}}\tilde{\sigma}^{-1}_{\mathbf{i}}\mathcal{S}_{[\mathbf{i}]} - \mathcal{S}_{[\mathbf{i}]}\right)^2 + 2\tau^2f_{\mathbf{i}}\tilde{\sigma}^{-1}_{\mathbf{i}}\mathcal{C}_{[\mathbf{i}]} + 2\tau^2f_{\mathbf{i}}\tilde{\sigma}^{-1}_{\mathbf{i}}\mathcal{S}_{[\mathbf{i}]}^2 \sum_{k=1}^K\frac{\frac{d}{d\sigma_{i_k}^k}f_{i_k}^k(\sigma_{i_k}^k)}{\sigma_{i_k}^kf_{i_k}^k(\sigma_{i_k}^k)}\right] - p\tau^2. \label{equation:sum.sure}
\end{align}
To update each $\lambda_k$, we simply apply a general purpose
univariate optimizer (e.g. Brent's method
\citep{brent1971algorithm}). To update $c$, we have
\begin{align*}
  &\frac{d}{dc}\left[c^2f_{\mathbf{i}}^2\tilde{\sigma}^{-2}_{\mathbf{i}}\mathcal{S}_{[\mathbf{i}]}^2 - 2 cf_{\mathbf{i}}\tilde{\sigma}^{-1}_{\mathbf{i}}\mathcal{S}_{[\mathbf{i}]}^2 + 2\tau^2cf_{\mathbf{i}}\tilde{\sigma}^{-1}_{\mathbf{i}}\mathcal{C}_{[\mathbf{i}]} + 2\tau^2cf_{\mathbf{i}}\tilde{\sigma}^{-1}_{\mathbf{i}}\mathcal{S}_{[\mathbf{i}]}^2 \sum_{k=1}^K\frac{1}{\sigma_{i_k}^kf_{i_k}^k(\sigma_{i_k}^k)}\right]\\
  &=2cf_{\mathbf{i}}^2\tilde{\sigma}^{-2}_{\mathbf{i}}\mathcal{S}_{[\mathbf{i}]}^2 - 2f_{\mathbf{i}}\tilde{\sigma}^{-1}_{\mathbf{i}}\mathcal{S}_{[\mathbf{i}]}^2 + 2\tau^2f_{\mathbf{i}}\tilde{\sigma}^{-1}_{\mathbf{i}}\mathcal{C}_{[\mathbf{i}]} + 2\tau^2f_{\mathbf{i}}\tilde{\sigma}^{-1}_{\mathbf{i}}\mathcal{S}_{[\mathbf{i}]}^2 \sum_{k=1}^K\frac{1}{\sigma_{i_k}^kf_{i_k}^k(\sigma_{i_k}^k)}.
\end{align*}
Let
\begin{align*}
  a &= \sum_{\mathbf{i}}f_{\mathbf{i}}^2\tilde{\sigma}^{-2}_{\mathbf{i}}\mathcal{S}_{[\mathbf{i}]}^2,\\
  b &= \sum_{\mathbf{i}}f_{\mathbf{i}}\tilde{\sigma}^{-1}_{\mathbf{i}}\mathcal{S}_{[\mathbf{i}]}^2,\\
  d &= \sum_{\mathbf{i}}\tau^2f_{\mathbf{i}}\tilde{\sigma}^{-1}_{\mathbf{i}}\mathcal{C}_{[\mathbf{i}]}, \text{ and}\\
  e &= \sum_{\mathbf{i}}\tau^2f_{\mathbf{i}}\tilde{\sigma}^{-1}_{\mathbf{i}}\mathcal{S}_{[\mathbf{i}]}^2 \sum_{k=1}^K\frac{1}{\sigma_{i_k}^kf_{i_k}^k(\sigma_{i_k}^k)},
\end{align*}
where we are summing over the set of $i_k$'s such that $\sigma_{i_k}^k
> \lambda_k$ for $k = 1,\ldots,K$. Then the minimum $c$ occurs at $(b
- d - e)/a$. This is a global minimizer, conditional on the
$\lambda_k$'s, since $a > 0$.

\section{General spectral functions}
\label{sec:gen.spec.func}
In Section \ref{subsection:diff}, we assumed that the spectral
functions were of the form:
\begin{align*}
  f^k(D_k) = \diag(f_1^k(\sigma_1^k),\ldots,f_{p_k}^k(\sigma_{p_k}^k)).
\end{align*}
That is, we only used $\sigma_i^k$ when determining the amount of
shrinkage to perform on $\sigma_i^k$. In this section, we will extend
these results to weakly differentiable functions of the form:
\begin{align*}
  f^k: \mathcal{D}_{p_k}^+ \rightarrow \mathcal{D}_{p_k}^+,
\end{align*}
where $\mathcal{D}_{p_k}^+$ is the space of $p_k$ by $p_k$ diagonal
matrices with non-negative diagonal elements. This will allow us to
use $\sigma_1^k,\ldots,\sigma_{p_k}^k$ to determine the amount of
shrinkage to perform on $\sigma_i^k$. These types of spectral
functions might be desirable if, for example, we wished to develop a
generalization of estimator (\ref{equation:improved.em}). Let
$\mathbf{s}_k = (\sigma_1^k,\ldots,\sigma_{p_k}^k)^T$ be the vector of
the $k$th mode specific singular values. We look at functions
\begin{align*}
  g^k: \mathbb{R}^{p_k+} \rightarrow \mathbb{R}^{p_k+},
\end{align*}
where $\mathbb{R}^{p_k+}$ is the space of $p_k$ vectors with
non-negative elements. Then
\begin{align*}
  f^k(D_k) = \diag(g^k(\mathbf{s}_k))
\end{align*}

The derivation of the SURE is the same as in Section
\ref{subsection:diff} except for the second term in
(\ref{equation:u.times.df}):
\begin{align*}
  \sum_{k=1}^K df(\tilde{D})_k[\Delta^{\bf i}] \cdot \mathcal{V}.
\end{align*}
We have:
\begin{align}
  \left(df(\tilde{D})_k[\Delta^{\bf i}] \cdot
    \mathcal{V}\right)_{[\mathbf{i}]} &= \left(\prod_{j\neq
      k}f_{i_j}^j(\sigma_{i_j}^j)\right)d(f^k \circ D_k)[\Delta^{\bf
    i}]_{[i_k,i_k]}\mathcal{V}_{[\mathbf{i}]}\nonumber\\
  &= \left(\prod_{j\neq k}f_{i_j}^j(\sigma_{i_j}^j)\right)d(g^k \circ \mathbf{s}_k)[\Delta^{\bf i}]_{[i_k]}\mathcal{V}_{[\mathbf{i}]} \label{equation:g.circ.s}
\end{align}

By the chain rule:
\begin{align*}
  d(g^k \circ \mathbf{s}_k)[\Delta^{\bf i}] = J_{g^k}(\mathbf{s}_k) d\mathbf{s}_k[\Delta],
\end{align*}
where $J_{g^k}(\mathbf{s}_k)$ is the Jacobian matrix of $g_k$
evaluated at $\mathbf{s}_k$. We know from (\ref{equation:d.comp}) that
\begin{align*}
  d\mathbf{s}_{k}[\Delta^{\bf i}]_{[j]} = 1(j = i_k)S_{[{\bf i}]}
  / \sigma_j^k \text{ for } j = 1,\ldots,p_k.
\end{align*}
So $d\mathbf{s}_k[\Delta^{\bf i}]$ contains zeros except in the $i_k$th position. Hence
\begin{align*}
  (J_{g^k}(\mathbf{s}_k) d\mathbf{s}_k[\Delta])_{[j]} = J_{g^k}(\mathbf{s}_k)_{[j,i_k]}S_{[{\bf i}]}
  / \sigma_{i_k}^k \text{ for } j = 1,\ldots,p_k
\end{align*}
And so
\begin{align}
  d(g^k \circ \mathbf{s}_k)[\Delta^{\bf i}]_{[i_k]} &=
  (J_{g^k}(\mathbf{s}_k) d\mathbf{s}_k[\Delta])_{[i_k]}\nonumber\\
  &= J_{g^k}(\mathbf{s}_k)_{[i_k,i_k]}S_{[{\bf i}]}/\sigma_{i_k}^k. \label{equation:jacob.single.element}
\end{align}
Inserting (\ref{equation:jacob.single.element}) into
(\ref{equation:g.circ.s}), we get:
\begin{align*}
  \left(df(\tilde{D})_k[\Delta^{\bf i}] \cdot
    \mathcal{V}\right)_{[\mathbf{i}]} &= \left(\prod_{j\neq
      k}f_{i_j}^j(\sigma_{i_j}^j)\right)
  J_{g^k}(\mathbf{s}_k)_{[i_k,i_k]}S_{[{\bf
      i}]}/\sigma_{i_k}^k\mathcal{V}_{[{\bf i}]}.
\end{align*}
That is, we only need the $(i_k,i_k)$th element of the Jacobian matrix
of the spectral function. Let
\begin{align*}
  J^k(D_k) =
  \diag(J_{g^k}(\mathbf{s}_k)_{[1,1]},\ldots,J_{g^k}(\mathbf{s}_k)_{[p_k,p_k]})
  \text{ for } k=1,\ldots,K.
\end{align*}
Then
\begin{align*}
  \sum_{k=1}^K df(\tilde{D})_k[\Delta^{\bf i}] \cdot \mathcal{V} = \sum_{k=1}^K Q_k \cdot \mathcal{S}^2
\end{align*}
where
\begin{align*}
  Q_k = (f^1(D_1)D_1^{-1},\ldots,f^{k-1}(D_{k-1})D_{k-1}^{-1},J_k(D_k)D_{k}^{-2},f^{k+1}(D_{k+1})D_{k+1}^{-1},\ldots,f^K(D_K)D_{K}^{-1}).
\end{align*}
The divergence is now of the form:
\begin{align*}
  \sumo\left(f(D) \cdot D^{-1}\cdot \mathcal{C} + \sum_{k=1}^KQ_k \cdot \mathcal{S}^2\right).
\end{align*}

\section{SURE for estimators that shrink elements in $\mathcal{S}$}
\label{sec:sure.s.shrink}
Consider the HOSVD (\ref{equation:hosvd}). In this section, we will
find the SURE for estimators of the form:
\begin{align}
  \label{equation:hose.s}
  t(\mathcal{X}) = U \cdot g(\mathcal{S}),
\end{align}
where
\begin{align*}
  (g(\mathcal{S}))_{[\mathbf{i}]} = g_{\mathbf{i}}(\mathcal{S}_{[\mathbf{i}]}).
\end{align*}
That is, we shrink each element of $\mathcal{S}$ separately. An
example of such a function is to soft-threshold each element of
$\mathcal{S}$:
\begin{align*}
  g_{\mathbf{i}}(\mathcal{S}_{[\mathbf{i}]}) = \sign(\mathcal{S}_{[\mathbf{i}]})(|\mathcal{S}_{[\mathbf{i}]}| - \lambda)_+,
\end{align*}
where $\sign(x)$ is $-1$ of $x < 0$, $1$ if $x > 0$, and $0$ if $x =
0$. Such a function induces $0$'s in the core array, which has
applications to increasing interpretability of higher-order PCA
\citep{henrion1993body,kiers1997uniqueness,murakami1998case,andersson1999general,de2001independent,martin2008jacobi}. Inducing
$0$'s in the core array is usually performed by applying orthogonal
rotations along each mode. Our approach provides an alternative
mechanism to induce $0$'s in the core array.

\begin{theorem}
  The differentials of $U_k$ and $\mathcal{S}$ are given in equations
  (\ref{equation:d.U}) and (\ref{equation:d.S}), respectively.
\end{theorem}
\begin{proof}
  We have already calculated $dU_k[\Delta]$ in Theorem
  \ref{theorem:diff}. To obtain $d\mathcal{S}[\Delta]$, we apply the
  chain rule to the HOSVD (\ref{equation:hosvd}) and solve for
  $d\mathcal{S}[\Delta]$.
\begin{align*}
  \Delta = d\mathcal{X}[\Delta] &= d(U\cdot \mathcal{S})[\Delta] =\sum_{k=1}^Kd\underline{U}_k[\Delta] \cdot \mathcal{S} + U \cdot d\mathcal{S}[\Delta],
\end{align*}
where $d\underline{U}_k[\Delta]$ is defined in (\ref{equation:U.underline}). Hence,
\begin{align}
  \label{equation:d.S}
  d\mathcal{S}[\Delta] =  U^T \cdot \Delta - \sum_{k=1}^K d\tilde{U}_k[\Delta]\cdot \mathcal{S}
\end{align}
where $d\tilde{U}_k[\Delta]$ is defined in (\ref{equation:U.tilde}).
\end{proof}

The derivation of the divergence for functions of the form
(\ref{equation:hose.s}) is very similar to that in Section
\ref{subsection:div}. The divergence may still be found from
(\ref{equation:divergence.index}). From the chain rule, we have:
\begin{align*}
  dt[\Delta^{\bf i}] = \sum_{k=1}^Kd\underline{U}_k[\Delta^{\bf i}] \cdot g(\mathcal{S}) + U \cdot d(g \circ \mathcal{S})[\Delta^{\bf i}],
\end{align*}
where this ``$\circ$'' means composition and
$d\underline{U}_k[\Delta^{\bf i}]$ is from
(\ref{equation:U.underline}). Hence,
\begin{align}
  \label{equation:u.times.df.S}
  U^T \cdot dt[\Delta^{\bf i}] = \sum_{k=1}^Kd\tilde{U}_k[\Delta^{\bf i}] \cdot g(\mathcal{S}) + d(g \circ \mathcal{S})[\Delta^{\bf i}],
\end{align}
where $d\tilde{U}_k[\Delta^{\bf i}]$ is from (\ref{equation:U.tilde}),
noting that the relationship in (\ref{equation:d.omega.i}) still
holds.

From the chain rule we have:
\begin{align*}
  d(f_{[\mathbf{i}]} \circ \mathcal{S}_{[\mathbf{i}]})[\Delta^{\bf i}]_{[\mathbf{i}]} = \left(\frac{d}{d\mathcal{S}_{[\mathbf{i}]}}f_{\mathbf{i}}(\mathcal{S}_{[\mathbf{i}]})\right) d\mathcal{S}_{[\mathbf{i}]}[\Delta^{\bf i}].
\end{align*}

We need the $(i_1,\ldots,i_K)$th element of
\begin{align}
  &\left(U^T \cdot df[\Delta^{\bf i}]\right)_{[\mathbf{i}]} \nonumber\\
  &= \left(\sum_{k=1}^Kd\tilde{U}_k[\Delta^{\bf i}] \cdot f(\mathcal{S}) + d(f \circ \mathcal{S})[\Delta^{\bf i}]\right)_{[\mathbf{i}]}\nonumber\\
  &= \sum_{k=1}^K\left(d\tilde{U}_k[\Delta^{\bf i}] \cdot f(\mathcal{S})\right)_{[\mathbf{i}]} + \left(\frac{d}{d\mathcal{S}_{[\mathbf{i}]}}f_{\mathbf{i}}(\mathcal{S}_{[\mathbf{i}]})\right) d\mathcal{S}_{[\mathbf{i}]}[\Delta^{\bf i}]\nonumber\\
  &= \sum_{k=1}^K\left(d\tilde{U}_k[\Delta^{\bf i}] \cdot f(\mathcal{S})\right)_{[\mathbf{i}]} + \left(\frac{d}{d\mathcal{S}_{[\mathbf{i}]}}f_{\mathbf{i}}(\mathcal{S}_{[\mathbf{i}]})\right) d\mathcal{S}[\Delta^{\bf i}]_{[\mathbf{i}]}\nonumber\\
  &= \sum_{k=1}^K\left(d\tilde{U}_k[\Delta^{\bf i}] \cdot f(\mathcal{S})\right)_{[\mathbf{i}]} + \left(\frac{d}{d\mathcal{S}_{[\mathbf{i}]}}f_{\mathbf{i}}(\mathcal{S}_{[\mathbf{i}]})\right) \left(\left(U^T \cdot \Delta^{\bf i}\right)_{[\mathbf{i}]} - \sum_{k=1}^K \left(d\tilde{U}_k[\Delta^{\bf i}]\cdot \mathcal{S}\right)_{[\mathbf{i}]} \right) \nonumber\\
  &= \sum_{k=1}^K\left(d\tilde{U}_k[\Delta^{\bf i}] \cdot f(\mathcal{S})\right)_{[\mathbf{i}]} + \left(\frac{d}{d\mathcal{S}_{[\mathbf{i}]}}f_{\mathbf{i}}(\mathcal{S}_{[\mathbf{i}]})\right) \left(E^{\bf i}_{[\mathbf{i}]} - \sum_{k=1}^K \left(d\tilde{U}_k[\Delta^{\bf i}]\cdot \mathcal{S}\right)_{[\mathbf{i}]} \right) \nonumber\\
  &= \sum_{k=1}^K\left(d\tilde{U}_k[\Delta^{\bf i}] \cdot f(\mathcal{S})\right)_{[\mathbf{i}]} + \left(\frac{d}{d\mathcal{S}_{[\mathbf{i}]}}f_{\mathbf{i}}(\mathcal{S}_{[\mathbf{i}]})\right) \left(1 - \sum_{k=1}^K \left(d\tilde{U}_k[\Delta^{\bf i}]\cdot \mathcal{S}\right)_{[\mathbf{i}]} \right).  \label{equation:final.of.right.simp}
\end{align}

Note that for any $\mathcal{A} \in \mathbb{R}^{p_1\times\cdots \times
  p_K}$
\begin{align*}
  &\left(d\tilde{U}_k[\Delta^{\bf i}] \cdot \mathcal{A}\right)_{[\mathbf{i}]} = \left((I_{p_1},\ldots,I_{p_{k-1}},d\Omega_{U_k}[\Delta^{\bf i}],I_{p_{k+1}},\ldots,I_{p_K})\cdot \mathcal{A} \right)_{[\mathbf{i}]}\\
  &= -\sum_{j=1, j\neq i_k}^{p_k}\mathcal{S}_{[i_1,\ldots,i_{k-1},j,i_{k+1},\ldots,i_K]}\mathcal{A}_{[i_1,\ldots,i_{k-1},j,i_{k+1},\ldots,i_K]}/[(\sigma_{i_k}^k)^2 - (\sigma_j^k)^2].
\end{align*}
Hence, from (\ref{equation:final.of.right.simp}) we have,
\begin{align}
  \diverge(g) &= \sum_{\bf i}\left[-\sum_{k=1}^K \sum_{j=1, j\neq i_k}^{p_k}\mathcal{S}_{[i_1,\ldots,i_{k-1},j,i_{k+1},\ldots,i_K]}f(\mathcal{S})_{[i_1,\ldots,i_{k-1},j,i_{k+1},\ldots,i_K]}/[(\sigma_{i_k}^k)^2 - (\sigma_j^k)^2]\right.\nonumber\\
  \nonumber&\left. + \left(\frac{d}{d\mathcal{S}_{[\mathbf{i}]}}f_{\mathbf{i}}(\mathcal{S}_{[\mathbf{i}]})\right) \left(1 + \sum_{k=1}^K \sum_{j=1, j\neq i_k}^{p_k}\mathcal{S}_{[i_1,\ldots,i_{k-1},j,i_{k+1},\ldots,i_K]}^2/[(\sigma_{i_k}^k)^2 - (\sigma_j^k)^2] \right)\right]  \\
  &= \sum_{\bf i}\left[-\sum_{k=1}^K \sum_{j=1, j\neq i_k}^{p_k}\frac{\mathcal{S}_{[i_1,\ldots,i_{k-1},j,i_{k+1},\ldots,i_K]}f_{[i_1,\ldots,i_{k-1},j,i_{k+1},\ldots,i_K]}(\mathcal{S}_{[i_1,\ldots,i_{k-1},j,i_{k+1},\ldots,i_K]})}{(\sigma_{i_k}^k)^2 - (\sigma_j^k)^2}\right. \label{div.formula1}\\
  \nonumber &\left. + \left(\frac{d}{d\mathcal{S}_{[\mathbf{i}]}}f_{\mathbf{i}}(\mathcal{S}_{[\mathbf{i}]})\right) \left(1 + \sum_{k=1}^K \sum_{j=1, j\neq i_k}^{p_k}\mathcal{S}_{[i_1,\ldots,i_{k-1},j,i_{k+1},\ldots,i_K]}^2/[(\sigma_{i_k}^k)^2 - (\sigma_j^k)^2] \right)\right].
\end{align}

We can rearrange the summations in the left part of
(\ref{div.formula1}) by switching the order of the $j$ and the $i_k$
and then altering the notation of the dummy variables to obtain:
\begin{align*}
  \diverge(g) &= \sum_{\bf i}\left[\mathcal{S}_{[\mathbf{i}]}f_{\bf i}(\mathcal{S}_{[\mathbf{i}]})\sum_{k=1}^K \sum_{j=1, j\neq i_k}^{p_k}1/[(\sigma_{i_k}^k)^2 - (\sigma_j^k)^2]\right. \\
  &\left. + \left(\frac{d}{d\mathcal{S}_{[\mathbf{i}]}}f_{\mathbf{i}}(\mathcal{S}_{[\mathbf{i}]})\right) \left(1 + \sum_{k=1}^K \sum_{j=1, j\neq i_k}^{p_k}\mathcal{S}_{[i_1,\ldots,i_{k-1},j,i_{k+1},\ldots,i_K]}^2/[(\sigma_{i_k}^k)^2 - (\sigma_j^k)^2] \right)\right].
\end{align*}

Hence, the SURE for these higher-order spectral functions
(\ref{equation:hose.s}) is:
\begin{align*}
  \sure(g(\mathcal{X})) &= -p\tau^2 + ||f(\mathcal{S}) - \mathcal{S}||^2 + 2\tau^2\sum_{\bf i}\left[\mathcal{S}_{[\mathbf{i}]}f_{\bf i}(\mathcal{S}_{[\mathbf{i}]})\sum_{k=1}^K \sum_{j=1, j\neq i_k}^{p_k}1/[(\sigma_{i_k}^k)^2 - (\sigma_j^k)^2]\right. \\
  &\left. + \left(\frac{d}{d\mathcal{S}_{[\mathbf{i}]}}f_{\mathbf{i}}(\mathcal{S}_{[\mathbf{i}]})\right) \left(1 + \sum_{k=1}^K \sum_{j=1, j\neq i_k}^{p_k}\mathcal{S}_{[i_1,\ldots,i_{k-1},j,i_{k+1},\ldots,i_K]}^2/[(\sigma_{i_k}^k)^2 - (\sigma_j^k)^2] \right)\right].
\end{align*}

\bibliography{sure_bib}
\end{document}